\documentclass[12pt]{elsarticle}

\usepackage{graphicx,amssymb,amsmath,amsthm,mathtools,subfigure,xspace}
\usepackage[ruled,vlined]{algorithm2e}
\usepackage[usenames,dvipsnames]{xcolor}
\usepackage{lineno}
\usepackage{booktabs}

\newtheorem{theorem}{Theorem}[section]
\newtheorem{corollary}{Corollary}[section]

\newtheorem{lemma}{Lemma}[section]
\newtheorem{observation}{Observation}[section]

\newcommand{\Obst}{{\sf Mutual Visibility}\xspace}
\newcommand{\Conv}{{\sf Convex Formation}\xspace}
\newcommand{\Circ}{{\sf Circle Formation}\xspace}
\newcommand{\Patt}{{\sf Pattern Formation}\xspace}
\newcommand{\Uni}{{\sf Uniform Circle Formation}\xspace}
\newcommand{\Near}{{\sf Near-Gathering}\xspace}
\newcommand{\Comm}{{\sf Communicating Vessels}\xspace}
\newcommand{\FSynch}{{\sc FSynch}\xspace}
\newcommand{\SSynch}{{\sc SSynch}\xspace}
\newcommand{\ASynch}{{\sc ASynch}\xspace}
\newcommand{\Rigid}{{\sc Rigid}\xspace}
\newcommand{\NRigid}{{\sc Non-Rigid}\xspace}

\newcommand{\Look}{{\it Look}\xspace}
\newcommand{\Compute}{{\it Compute}\xspace}
\newcommand{\Move}{{\it Move}\xspace}
\newcommand{\NearG}{{\sf Near Gathering}\xspace}

\DeclarePairedDelimiter{\norm}{\lVert}{\rVert}
\DeclarePairedDelimiter{\abs}{\lvert}{\rvert}

\begin{document}

 \begin{frontmatter}
\title{Mutual Visibility by Luminous Robots\\Without Collisions}

 \author[roma]{G.A.~Di Luna}
 \ead{diluna@dis.uniroma1.it}
\author[ottawa]{P.~Flocchini}
\ead{flocchin@site.uottawa.ca}
\author[kolkata]{S.~Gan Chaudhuri}
\ead{srutiganc@it.jusl.ac.in}
\author[pisa]{F.~Poloni}
\ead{fpoloni@di.unipi.it}
\author[carleton]{N.~Santoro}
\ead{santoro@scs.carleton.ca}
\author[ottawa]{G.~Viglietta}
\ead{viglietta@gmail.com}

\address[roma]{Dipartimento di Ingegneria Informatica, Automatica e Gestionale, Universit\`a degli Studi di Roma ``La Sapienza'', Italy}
\address[ottawa]{School of Electrical Engineering and Computer Science, University of Ottawa, Canada}
\address[kolkata]{Department of Information Technology, Jadavpur University, Kolkata, India}
\address[pisa]{Dipartimento di Informatica, Universit\`a di Pisa, Italy}
\address[carleton]{School of Computer Science, Carleton University, Ottawa, Canada}


\date{}

%
%

\begin{abstract}
 {\footnotesize
Consider a finite set of identical computational entities that can move freely in the Euclidean plane operating
in  Look-Compute-Move cycles.
Let $p(t)$ denote the location of entity $p$ at time $t$; entity $p$ can see entity $q$ at time $t$ if at that time no other entity lies on the line segment $p(t)q(t)$. We consider the basic problem called \Obst: starting from arbitrary distinct locations, within finite time the entities must reach, without collisions, a configuration where they all see each other. This problem must be solved by each entity autonomously executing the same algorithm.
 We study this problem in the {\em  luminous robots} model; in this  generalization of the standard model of oblivious robots,  
 each entity, called {\em robot},  has  an externally  visible persistent  light that can  assume colors from a fixed set of size $c$.
The case where the number of 
colors is less than 2 (i.e., $c\leqslant 1$) corresponds to the classical model without lights: indeed, having lights of one possible color is equivalent to having no lights at all.
 
The extensive literature on computability in such a model, mostly for $c\leqslant 1$ and recently for $c>1$,
has never considered the problem of \Obst because it has always assumed that three collinear robots are mutually visible. 

In this paper we remove this assumption, and investigate 
under what conditions luminous robots can solve \Obst without collisions, and at  
what cost, in terms of the number of colors used by the robots. 
We establish a spectrum of results, depending on the power of the adversary (i.e., the scheduler controlling the robots' actions),
on the number $c$ of colors,
and on the a-priori knowledge the robots have about the system.
Among such results,  we prove that 
  \Obst can  always be solved without collisions
in \SSynch  with $c=2$ colors   and   in    \ASynch with  $c=3$ colors.
If an adversary  can interrupt and stop a robot before it reaches its computed destination,
\Obst is still solvable without collisions   in  \SSynch with $c=3$ colors,
and, if the robots agree on the direction of one axis,
also  in  \ASynch.
 All the results are obtained constructively by means of novel
protocols. 

As a byproduct of our solutions, 
  we provide  the first
obstructed-visibility  solutions to two  classical problems for oblivious robots: 
{\em collision-less convergence to a point} (also called \emph{near-gathering})
and {\em  circle formation}.
  }
\end{abstract}

\end{frontmatter}


\section{Introduction}\label{s:introduction}

\subsection{Computational  Framework}

Consider a 
distributed system  composed of a team 
of mobile computational entities, called
 {\em robots},  moving and operating
  in the Euclidean plane $\mathbb R^2$, initially each at  a distinct point.  
Each robot can move freely in the plane, and
operates  in 
  \Look-\Compute-\Move  cycles.
During a cycle, a robot determines the position (in its own coordinate system) of the other robots
(\Look); it 
executes a protocol (which is deterministic and it is the same for all robots)  to determine a destination point
 (\Compute); and
moves towards the computed destination (\Move). After each cycle, a robot may be inactive for an arbitrary but finite
amount of time.
The robots are  {anonymous},
 without a central control, and  {oblivious} (i.e., at the beginning of a cycle,  a robot has no memory of any observation or computation performed in
its previous cycles). 
What is computable by such entities  has  been the object of extensive research within distributed computing;
e.g., see~\cite{AgP06, CiFPS12, CohP05, DaFSY14,EfP07,FloPSW08,FujYKY12,ISKIDWY12,SuS90,SuY99,YaS10};  
 for a recent review see~\cite{FlPS12}.

%
%
%
%

   Vision and mobility provide the robots with {\em stigmergy}, enabling
    the robots to  communicate and coordinate their actions   by moving
and sensing their relative positions; they are
 otherwise assumed to lack any
   means of explicit direct  communication.
   This restriction could enable   deployment in extremely harsh
environments   where  communication is
impossible or can be jammed. Nevertheless, in many  other situations
it is possible to assume the availability of some sort of direct
communication.
 The theoretical interest is obviously for weak
communication capabilities.

A model employing a  weak explicit communication mechanism is that of
{\em  robots with lights}, or {\em luminous  robots},
initially suggested by Peleg~\cite{Peleg2005}.
   In this model,
    each robot is provided with  a local  externally-visible {\em  light},
which can  assume colors from a fixed set.  The
robots  explicitly communicate  with each other using these lights.
The lights  are persistent (i.e., the color is not erased at the end
of a cycle), but otherwise
the robots are {oblivious}~\cite{DasFPSY12,DasFPSY14,EfP07,FlSVY13,Peleg2005,Vi13}.
Notice that a light with only one possible color is the same as no light; hence
the luminous robots model generalizes the classical one.

Both in the classical model and in that with lights, 
depending on the assumptions on   the activation schedule  and the
duration of the cycles, different  settings are identified. 
 In the synchronous setting,
the robots operate in rounds, and 
all the robots that are activated in a round perform their  cycle in perfect synchrony. In this case,
the system is {\em fully synchronous} (or \FSynch) if all robots are activated at all rounds, and it is {\em semi-synchronous} (or \SSynch) otherwise.
In the {\em asynchronous} setting (or \ASynch),
there is no common notion of time,  and no assumption is made on the
timing and duration of each computation and movement, other than that it is finite. (For the \SSynch and \ASynch models there are bland fairness assumptions that prevent robots from remaining inactive forever, which are discussed in Section~\ref{s:robotmodels}).

The choice of when a robot is activated (in \SSynch) and the duration of an activity within a cycle
(in \ASynch) is made under the control of an {\em adversary}, or \emph{scheduler}.
Similarly, the choices of the initial location of each robot and of its 
private coordinate system are made under adversarial conditions.

A crucial distinction is whether or not the adversary  
has also the power to  stop
a moving robot before it reaches its destination. 
If so,  the moves are said to be  {\em non-rigid}. The only constraint is that, if interrupted
before reaching its destination, the robot moves  at least a
  minimum distance  $\delta>0$
(otherwise, the adversary would be able to prevent robots from reaching any destination, in any amount of cycles).  
If the adversary does not have such a power, the moves are
said to be {\em rigid}. The model with rigid moves is referred to as \Rigid, and the other one is called \NRigid.

In the rest of the paper, with abuse of terminology, we will often refer to \FSynch, \SSynch, or \ASynch robots or schedulers (as opposed to systems), and to \Rigid or \NRigid robots or schedulers  (as opposed to models).


%

\subsection{Obstructed Visibility}

The classical model   and  the more recent model of
 robots with lights share a common assumption: that 
 three or more collinear robots
   are 
mutually visible.
It can be easily argued against such an assumption, and for the
importance of investigating computability  when  visibility is {\em obstructed} 
by the presence of other robots: that is, if two robots $r$ and $s$ are located at 
$r(t)$ and $s(t)$ at time $t$, they can see each other if and only 
if no other robot lies on the segment $r(t)s(t)$ at that time. 

Very little is known   on computing with {obstructed visibility}.
 In fact, the few studies on obstructed visibility have been carried out in other models: 
 the model of robots in the {\em one-dimensional} space $\mathbb R$~\cite{CoP08}; and
 the so-called \emph{fat robots}
model, where  robots are  not geometric points but occupy unit disks, and collisions are
   allowed and can be  used as an explicit computational tool
 (e.g.,~\cite{AgGM13,BoKF12,CzGP09}).
 In our model,  collisions  can create
unbreakable symmetries:
 since robots are oblivious and anonymous and execute the same protocol, 
 if $r(t)=s(t)$ (a collision),  then the activation adversary can force
$r(t')=s(t')$ for all $t'>t$ if the two robots do not have lights or    their lights have the same color.
Thus, unless this is the intended outcome,
 collision avoidance is always a requirement 
for all algorithms in the model considered here.
%


In this paper we   focus on luminous robots  in the presence of    obstructed visibility, and investigate 
 computing in such a setting.
Clearly, obstructed visibility increases the difficulty of solving problems
without the use of additional assumptions.
For example,  with unobstructed visibility, every active robot can
determine the  total
number $n$ of robots at each activity cycle.
With obstructed visibility,
unless a robot has a-priori knowledge of  $n$ 
  and this knowledge is persistently stored, 
 the robot 
  might be unable
to decide  if it sees all the robots;  hence it might 
be unable to determine the value $n$.

%

The main problem we investigate,  called \Obst,  is perhaps the most basic in a situation of
obstructed visibility:  starting from arbitrary distinct positions in the plane,
 within finite time  the robots
 must reach  a configuration in which they are in distinct locations,
 they can all see each other, and they no longer move.
This problem is clearly at the basis of any subsequent task requiring complete visibility.
Notice that this problem does not exist under unobstructed visibility, and  has never been
investigated before.

Among the configurations that achieve mutual visibility, a special
class is that where all robots 
 are in a strictly convex position; within that class,
 of particular interest are those where the robots
 are on the perimeter of a  circle, possibly  equally spaced.
The problems of forming  such
configurations  (respectively called \Conv and \Circ)
 have been extensively studied
both directly
   (e.g.,~\cite{DatDGM13,Defago2008,DiLaPe08,FlPSV14})
and as part of the more general \Patt problem  (e.g.,~\cite{FloPSW08,FujYKY12,SuS90,SuY99,YaS10}).
Unfortunately, none of these
investigations consider obstructed visibility, and those algorithms do
not work in the setting considered here.

Note that a requirement of the \Obst problem is that robots stop moving after they have reached a configuration in which they all see each other. To this end, we will grant robots the ability to perform a special operation called \emph{termination}, after which they can no longer be activated by the scheduler. The termination operation is especially useful in practice when the robots have to perform several tasks in succession. Of course, even if this operation is not directly available, it can still be simulated via the addition of an extra color, which can be used by a robot to indicate (to the other robots, as well as to itself) that it has terminated. Moreover, if the termination operation is removed from the algorithms presented in this paper (and some straightforward adjustments are made, which do not require extra colors), then a weaker form of the \Obst problem is solved, in which the robots get to permanently see each other, but they never stop moving. That is, all obstructions are permanently removed, but the termination condition is not met. In some cases, removing the termination operation will even allow us to successfully apply our algorithms to different problems, such as \Circ and \NearG, as discussed in Section~\ref{s:extra}.

\subsection{Main Contributions}


In this paper we investigate 
under what conditions luminous robots can solve \Obst and at  
what cost (i.e., with how many colors). We establish a spectrum of results, depending on the power of the adversary,
on the number $c$ of colors,
and on the a-priori knowledge the robots have about the systems.

We first consider the case when the adversary 
can choose the activation
schedule (in \SSynch) and the duration of each robot's operations (in \ASynch), but 
cannot interrupt the movements of the
robots; that is, movements are rigid.
In this case, we show the following.

\begin{theorem}\label{t:sum1}
\Obst is solvable without collisions by \Rigid robots
\begin{itemize}
\item[\emph{\textbf{(a)}}] with  {\em no} colors  in  \SSynch,
if the robots know their number, $n$;
\item[\emph{\textbf{(b)}}] with $2$ colors  in  \SSynch, always;
\item[\emph{\textbf{(c)}}] with $3$ colors  in  \ASynch, always.
\end{itemize}
\end{theorem}

We then consider the case when the adversary has  also the
power to interrupt the movements of the
robots; that is, movements are non-rigid.
The only restriction is that there exists a constant absolute length  $\delta>0$ such that, even if a robot's move is interrupted before it reaches the destination, it travels at
least a length $\delta$ towards it (otherwise it many never be able to reach any destination).
In the case of non-rigid movements, we prove the following.

\begin{theorem}\label{t:sum2}
\Obst is solvable without collisions by  \NRigid  robots
\begin{itemize}
\item[\emph{\textbf{(a)}}] with \emph{no} colors in  \SSynch, if the robots know $\delta$ and their number, $n$;
\item[\emph{\textbf{(b)}}] with $2$ colors in \SSynch, if the robots know  $\delta$;
\item[\emph{\textbf{(c)}}] with $3$ colors  in  \SSynch, always;
\item[\emph{\textbf{(d)}}] with $3$ colors  in  \ASynch, if the robots agree on the direction of one  coordinate axis.
\end{itemize}
\end{theorem}

All these results are  established constructively. 
We present  and analyze two  protocols,
 Algorithm~\ref{alg1} ({\em Shrink}) and Algorithm~\ref{alg2} ({\em Contain}),
whose goal is to allow the robots to position themselves at the vertices of a  convex  polygon,
solving  \Conv, and thus \Obst. 
These two algorithms are 
based on  different strategies, and are tailored for different situations.
Protocol {\em Shrink}  uses two colors and requires  rigid movements, while 
protocol {\em Contain} uses more colors but operates also with non-rigid
movements. We prove their correctness for \SSynch robots (Sections~\ref{s:ssynch} and~\ref{s:ssynch2}). 
We then show how,
directly or with simple expansions and modifications of these two algorithms, all the claimed results
follow (Sections~\ref{s:asynch} and~\ref{s:extra}). Finally, we propose some open problems (Section~\ref{s:conclusions}).

Let us point out that, to prove the correctness of {\em Shrink},
we solve a seemingly unrelated problem, \Comm, which is interesting in its own right.

As a byproduct of our solutions, 
  we provide  the first
obstructed-visibility  solution to a 
 classical problem for oblivious robots:
{\em collision-less convergence to a point} (\Near) (see~\cite{FlPS12,PaPeVi12}),
Indeed, 
if the robots continue to follow  algorithm {\em Shrink} once they reach full visibility, 
the convex hull of their positions converges to a point, and the robots approach it without colliding,
  thus solving \Near (Section~\ref{s:nearg}). This  algorithm has  an
interesting 
 fault-tolerance property: if a single robot is faulty and becomes unable to move, the robots will 
 still solve \Near, converging to the faulty robot's location (Section~\ref{s:fault}).
 
Additionally, both protocols can be modified so 
that the robots can position themselves  on the perimeter of a  circle, thus providing an
obstructed-visibility  solution to
the classical problem of  \Circ. The problem can be solved 
with $2$ colors in  \Rigid\ \SSynch, 
with $3$ colors  in \NRigid\ \SSynch,
and
with $4$ colors in \Rigid\ \ASynch, and \NRigid\ \ASynch with agreement on one axis (Section~\ref{s:circle}).

\section{Model and Definitions}
\subsection{Modeling Robots}\label{s:robotmodels}
 We mostly follow the terminology and definitions of the standard model of oblivious mobile robots (e.g., see~\cite{FlPS12}).

By ${\cal R}=\{r_1 ,r_2, \cdots, r_n\}$ we denote  a set of oblivious   mobile computational entities, called {\em robots},
 operating in the Euclidean plane, and initially placed at distinct points.
 Each robot is provided with its own local coordinate system centered in itself, and its own notion of unit distance and handedness. 
We denote by $r(t) \in \mathbb{R}^{2}$ the position occupied by robot $r\in {\cal R}$ at time $t$; these positions are expressed here in a global coordinate system, which is used for description purposes, but is unknown to the robots. 
Two robots $r$ and $s$ are said to {\em collide} at time $t$ if $r(t)=s(t)$.
A robot $r$ can \emph{see} another robot $s$ (equivalently, $s$ is \emph{visible} to $r$) at time $t$ if and only  if no other robot lies in the segment $r(t)s(t)$ at that time. 

The robots are luminous:
 each robot $r$ has a persistent state variable, called {\em light},
which may assume any value  
 in a finite set $\mathcal C$ of {\em colors}. 
 The color of $r$ at time $t$ can be seen by all robots 
that can see $r$ at that  time.

The robots are autonomous  (i.e., without any external control), anonymous (i.e., without internal identifiers), indistinguishable (i.e., without external markings), without any direct means of communication, other than their lights.  At any time, robots can be performing a variety of operations, but initially (i.e., at time $t=0$) they are all still and idle.


When activated, a  robot  performs  a  \emph{Look-Compute-Move}  sequence of operations:
 it first obtains a snapshot 
 of the positions, expressed in its local coordinate system,   of all visible robots, along with their respective colors  ({\em Look} phase);
 using the last obtained snapshot as an input,
the robot    executes a deterministic algorithm, which is the same for all robots, to compute
 a destination point $x \in \mathbb{R}^{2}$ (expressed in its local coordinate system)
 and a color $c \in \mathcal C$, and it   
 sets  its light  to $c$
   ({\em Compute} phase); 
  finally, it     
 moves  towards $x$ ({\em Move} phase).
It then starts a new cycle, whenever the scheduler (which is an abstract entity controlling to some extent the behavior of the robots) decides to activate it again.
%

In the \emph{Compute} phase, a robot may also decide to terminate its execution. When a robot has terminated, it remains still forever, and its light remains the same color that it was at the moment of termination.

The robots are {\em oblivious} in the sense that, 
when a robot transitions from one cycle to the next, all its local memory, except for the light, is reset.
In other words,
a robot  has no  memory of  past computations and snapshots, except for the light.

With regards to the activation and timing of the robots, there are two basic settings:
{\em semi-synchronous} (\SSynch) and  {\em  asynchronous} (\ASynch).
In  \SSynch,  the time is discrete; at each time instant $t\in \mathbb N$ (called a {\em round} or a \emph{turn}) 
 a subset of the robots is activated by the scheduler and performs a whole \emph{Look-Compute-Move} cycle atomically and instantly.
 At any given round,  any subset of robots may be activated, from the empty set to all of ${\cal R}$. In particular, if all robots are activated at every round, the setting is called {\em fully synchronous} (\FSynch). There is a bland \emph{fairness} constraint on the choices that the scheduler can make: every robot must be activated infinitely many times (unless it terminates).
 In  \ASynch,  there is no common notion of time: each robot executes its cycles independently, the \emph{Look} operation is   instantaneous, but
   the   \emph{Compute} and \emph{Move} operation can take an unpredictable (but finite) amount of time, unknown to the robot. In a \emph{Move} phase there are no constraints on the speed of a robot, as long as it always moves directly towards its destination point at non-negative speed.

The scheduler that controls the activations (in \SSynch) and the durations of the operations (in \ASynch) can be thought of as an {\em adversary}, whose purpose is to prevent the robots from doing their task. Other than acting as a scheduler, the adversary also determines the initial position of the robots and their local
 coordinate systems; in particular,  the coordinate system of a robot might not be preserved over time and might 
be modified by the adversary between one cycle and the next. In the simplest model, the robots do not necessarily agree on the orientation of the coordinate axes, on the unit distance, and on the clockwise direction (i.e., the handedness of the system). However, in Section~\ref{sec:axis}, we will discuss the special model in which all the robots agree on the direction of one axis, and the adversary is unable to change it.

 The adversary might or might not have the power  to interrupt the
movement  of a robot  before it  reaches its destination in the {\em Move} operation.
If it does,  the system  is said to be \NRigid. The only constraint on the adversary is that there exists
a constant  $\delta>0$ such that, if interrupted 
before reaching its destination, a robot moves  at least a  distance 
  $\delta$. The value of $\delta$ is decided by the scheduler once and for all, and typically it is not known by the robots, which therefore cannot use it in their computations (we will discuss the scenario in which the robots know the value of $\delta$ in Section~\ref{s:delta}).
   Notice that, without this constraint, the adversary would be able to prevent a robot from reaching any given destination in a finite number of turns.
 If movements are not under the control of the adversary, and every robot reaches its destination at every turn, the system is said to be \Rigid.


%


\subsection{Mutual Visibility and Related Problems}

 The \Obst problem requires $n$  robots to form a configuration in which
  they occupy $n$ distinct locations, and no three of them are collinear. Subsequently, the robots have to terminate.
 A protocol $P$ is a solution of  \Obst if it allows 
  the robots to solve \Obst
starting   
    from any initial configuration 
 in which  their positions are all distinct, 
and regardless of the decisions of
   the adversary (including the activation schedule, the local coordinate systems of the robots, and the value of $\delta$).   
   
Let us stress that,
   since robots are oblivious and anonymous and execute the same protocol, 
 if $r(t)=s(t)$ (a collision),  then the  adversary can force
$r(t')=s(t')$ for all $t'>t$ if the two robots do not have lights or   their lights have the same color.
Hence the two robots will never again  occupy distinct locations, and will no longer be able to
solve \Obst. 
Thus,   collision avoidance of robots with the same color  is  a requirement 
for any solution protocol.
   
 Among the configurations that solve the \Obst problem, a special
class is that in which all robots 
 are in a strictly convex position. Within this class,
 of particular interest are the configurations in which the robots
lie on the perimeter of a  circle. Among these,
there are the notable configurations in which the robots 
occupy the vertices of a regular $n$-gon.
The problems of forming  such
configurations  are called  \Conv , \Circ, and 
\Uni, respectively.

\subsection{Geometric Notions and Observations}
A finite set of points $S\subset \mathbb R^2$ is said to be \emph{convex} if all the points of $S$ lie on the perimeter of the convex hull of $S$. If a point $p$ of a convex set $S$ lies in the relative interior of an edge of the convex hull of $S$, then $p$ is said to be a \emph{degenerate} vertex of the convex hull. If none of the points of a convex set $S$ is a degenerate vertex of the convex hull, then $S$ is said to be a \emph{strictly convex} set. On the other hand, we will say that a polygon is \emph{degenerate} if its area is zero. These two notions of degeneracy are used in different contexts (one refers to vertices, the other refers to whole polygons), hence they can hardly be confused.

Let $\mathcal H(t)$ denote the convex hull of $\{r_1(t), r_2(t), \cdots, r_n(t)\}$ at time $t$.
The robots lying on its boundary are called \emph{external robots} at time $t$, while the ones lying in its interior are the \emph{internal robots} at time $t$.

Observe that a robot may not know where the convex hull's vertices are located, because its view  may be obstructed by other robots. However, it can easily 
determine whether it is an external or an internal robot.  In fact,   a robot $r $ is {external} at time $t$ if and only if there is a half-plane bounded by a straight line through $r(t)$ whose interior contains no robots at time $t$. In other words, $r$ is external if and only if it lies on the boundary of the convex hull of the robots that it can currently see. Note also that the neighbors of an external robot on its visible convex hull are indeed its neighbors on the actual convex hull. If, in addition, $r$ lies at a non-degenerate vertex of the (visible) convex hull, it is said to be a {\em vertex} robot.

Moreover, a robot is able to tell if $\mathcal H$ is a line segment, i.e., if all the robots are collinear. In particular, if a robot can see only one other robot, it understands that it is an \emph{endpoint robot}. Conversely, non-endpoint robots can always see more than one other robot.
 
The points of $\mathbb R^2$ are treated like vectors, and as such they can be added, subtracted, multiplied by scalars, etc. The dot product between vectors $a$ and $b$ will be indicated by the expression $a\bullet b$.

\section{Solving \Obst for \Rigid\ \SSynch Robots}\label{s:ssynch}


In this section we consider the   \Obst problem  in  the \Rigid\   \SSynch setting. 
We present and analyze  a  protocol,
 Algorithm~\ref{alg1} ({\em Shrink}),  and we prove it
  solves \Obst in such a setting using   only two   colors.

\subsection{Description of Algorithm~\ref{alg1}}\label{s:algorithm}

The main idea of Algorithm~\ref{alg1} is to make only the external robots move, so as to shrink the convex hull. When a former internal robot becomes external, it starts moving as well. Eventually, all the robots reach a strictly convex configuration, and at this point they all see each other and they can terminate.

\SetAlgoCaptionSeparator{: Shrink.}
\begin{algorithm}\label{alg1}
\caption{Solving the \Obst problem for \Rigid\ \SSynch robots with 2-colored lights}
\DontPrintSemicolon
\LinesNumbered
\KwIn{$\mathcal V$: set of robots visible to me (myself included) whose positions are expressed in a coordinate system centered at my location.}
$r^* \longleftarrow$ myself\;
$\mathcal P \longleftarrow \{r.\mbox{\emph{position}}\mid r\in \mathcal V\}$\;
$\mathcal H \longleftarrow$ convex hull of $\mathcal P$\;
\If{$|\mathcal V|=3$ \textbf{\emph{and}} $\mathcal H$ is a line segment}{
	Move orthogonally to $\mathcal H$ by any positive amount\;
}
\Else{
\If{$r^*.\mbox{position}$ is a vertex of $\mathcal H$}{
	$r^*.\mbox{\emph{light}} \longleftarrow \mbox{\emph{Vertex}}$\;
	\lIf{$\forall r\in \mathcal V,\, r.\mbox{light} = \mbox{Vertex}$}{Terminate}
	\ElseIf{$|\mathcal V|>2$}{
		$a \longleftarrow$ position of my ccw neighbor on the boundary of $\mathcal H$\;
		$b \longleftarrow$ position of my cw neighbor on the boundary of $\mathcal H$\;
		$u\longleftarrow a/2$\;
		$\gamma \longleftarrow 1/2$\;
		\ForEach{$r\in \mathcal V\setminus\{r^*\}$}{
			Let $\alpha$, $\beta$ be such that $r.\mbox{\emph{position}}=\alpha\cdot a+\beta\cdot b$\;
			\If{$\alpha+\beta<\gamma$}{
				$u\longleftarrow r.\mbox{\emph{position}}$\;
				$\gamma \longleftarrow\alpha+\beta$\;
			}
			\lElseIf{$\alpha+\beta=\gamma$ \emph{\textbf{and}} $r.\mbox{position}$ is closer to $b$ than $u$}{$u\longleftarrow r.\mbox{\emph{position}}$}
		}
		$v \longleftarrow \gamma\cdot b$\;
		Move to $(u+v)/2$\;
	}
}
\lElseIf{$\forall r\in \mathcal V\setminus\{r^*\},\,r.\mbox{light} = \mbox{Vertex}$ \emph{\textbf{and}} $r^*.\mbox{position}$ lies in the interior of $\mathcal H$}{Move to the midpoint of any edge of $\mathcal H$
}}
\end{algorithm}

If an active robot $r_i$, located at $p$, realizes that it is not a vertex robot, it does not move. Otherwise, it locates its clockwise and counterclockwise neighbors (in its own coordinate system) on the convex hull's boundary, say located at $a$ and $b$, which are necessarily visible. Then, $r_i$ attempts to move somewhere in the triangle $\triangle pab$, in such a way to shrink the convex hull, and possibly make one more robot become a vertex robot. To avoid collisions with other robots that may be moving at the same time, $r_i$'s movements are restricted to a smaller triangle, shaded in gray in Figure~\ref{fig1}. Moreover, to avoid becoming a non-vertex robot, $r_i$ does not cross any line parallel to $ab$ that passes through another robot, and it carefully positions itself on the closest of such lines, as shown in Figure~\ref{fig1:a}. In particular, if no such line intersects the gray area, $r_i$ makes a \emph{default move}, and it moves halfway toward the midpoint of the segment $ab$, as indicated in Figure~\ref{fig1:b}.

\begin{figure}[h]
\centering
\subfigure[Making $c$ become a vertex robot, without moving past it]{\label{fig1:a}\includegraphics[width=.8\linewidth]{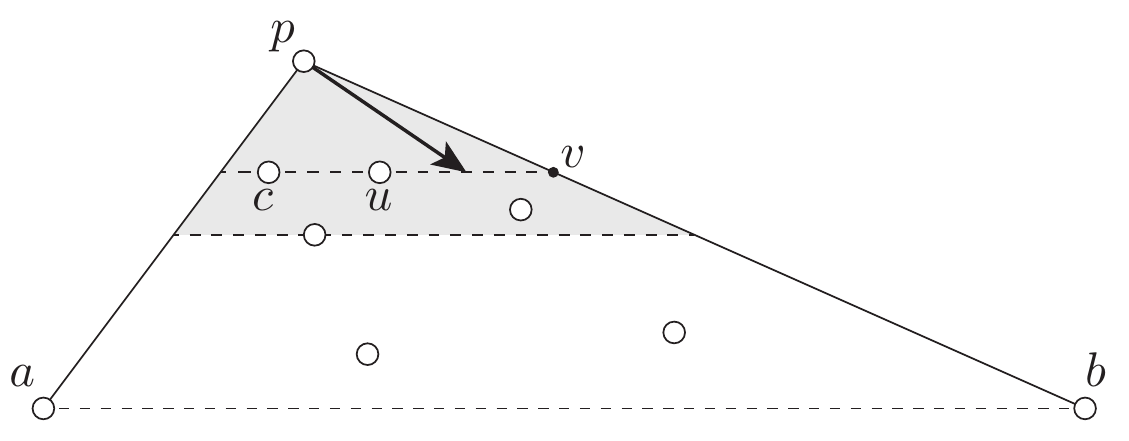}}
\subfigure[Default move]{\label{fig1:b}\includegraphics[width=.8\linewidth]{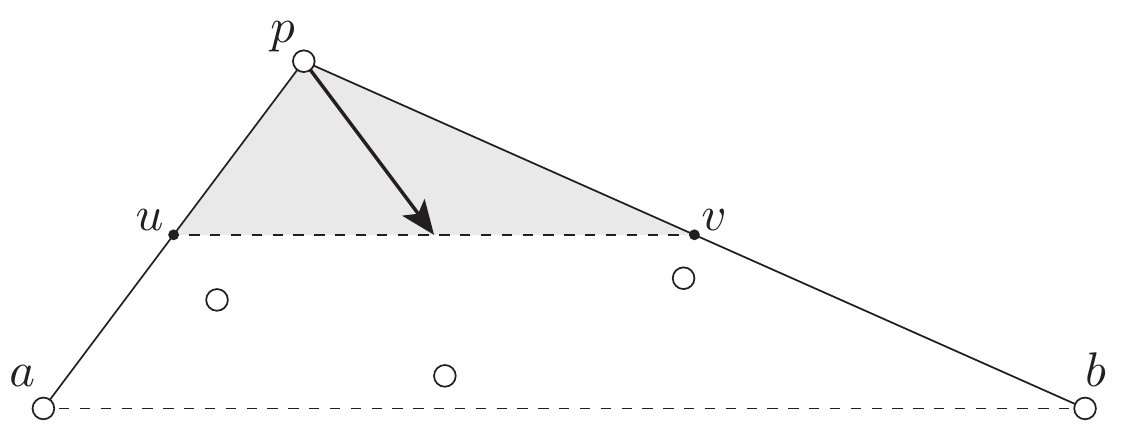}}
\caption{Move of an external robot, in two different cases (robots' locations are indicated by small circles)}
\label{fig1}
\end{figure}

In order to recognize that the \Obst problem has been solved, and to correctly terminate, the robots carry visible lights of two possible colors: namely, $\mathcal C=\{\mbox{\emph{Off}},\mbox{\emph{Vertex}}\}$. All robots' lights are initially set to \emph{Off}. If an active robot realizes that it is a vertex of the convex hull, it sets its light to the other value, \emph{Vertex}. Hence, when a robot sees only robots whose lights are set to \emph{Vertex}, it knows it can see all the robots in the swarm, and hence it terminates.

The above rules are sufficient to solve the \Obst problem in most cases, but there are some exceptions. It is easy to see that there are configurations in which \Obst is never solved until an internal robot moves, regardless of the algorithm employed. For instance, suppose that the configuration is centrally symmetric, with one robot lying at the center. Let the local coordinate systems of any two symmetric robots be oriented symmetrically and have the same unit distance, and assume that the scheduler chooses to activate all robots at every turn. Then, every two symmetric robots have symmetric views, and therefore they move symmetrically. If the central robot---which is an internal robot---never moves, then the configuration remains centrally symmetric, and the central robot always obstructs all pairs of symmetric robots. Hence \Obst is never solved, no matter what algorithm is executed.

It turns out that our rules can be fixed in a simple way to resolve also this special case: whenever an internal robot sees only robots whose lights are set to \emph{Vertex} (except its own light), it moves to the midpoint of any edge of the convex hull.

Finally, the configurations in which all the robots are collinear need special handling. In this case it is impossible to solve \Obst unless some robots leave the current convex hull. Suppose that a robot $r$ realizes that all robots lie on a line, and that it is not an endpoint (i.e., $r$ can see only two other robots, which are collinear with it). Then, $r$ moves by any positive amount, orthogonally to the line formed by the other two visible robots. When this is done, the previous rules apply.

\subsection{Correctness of Algorithm~\ref{alg1}}
\subsubsection{Invariants}\label{s:invariants}

In the following we discuss some basic invariants, which will serve to prove the correctness of Algorithm~\ref{alg1}.

Suppose that, for some $t\in\mathbb{N}$, $\mathcal{H}(t)$ is not a line segment: the situation is illustrated in Figure~\ref{fig2}. If a vertex robot is activated, it is bound to remain in the corresponding gray triangle, called \emph{movement region} of the robot. More precisely, the movement region consists of the interior of the gray triangle, plus the vertex where the robot currently is, plus the interior of the edge that is opposite to the robot. Hence all movement regions are disjoint. Moreover, if there is only one internal robot and it sees only robots whose light is set to \emph{Vertex}, it moves to the midpoint of an edge of $\mathcal H(t)$, which does not lie in any movement region. It follows that, no matter which robots are activated at time $t$, they will not collide at time $t+1$. Also, $\mathcal H(t+1)\subseteq\mathcal H(t)$.

\begin{figure}[h]
\centering
\includegraphics[width=.8\linewidth]{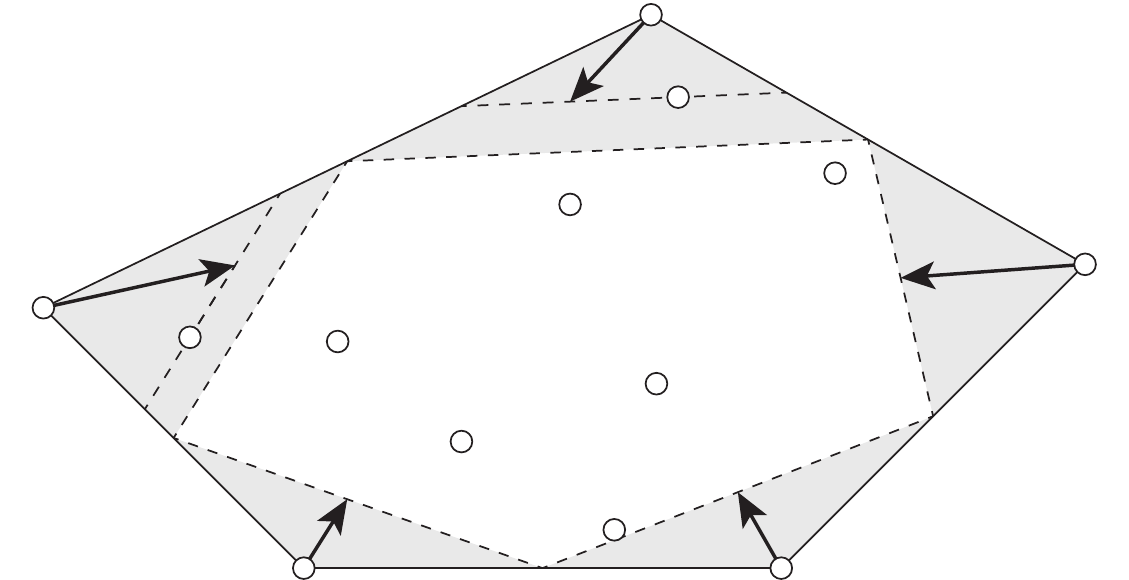}
\caption{Combined motion of all vertex robots}
\label{fig2}
\end{figure}

Recall that a robot $r\in\mathcal{R}$ is a vertex robot if an only if it lies at the vertex of a reflex angle whose interior does not contain any robots. Now, referring to Figure~\ref{fig1}, it is clear that a vertex robot will remain a vertex robot after a move. 
Additionally, if no new vertex robots are acquired between time $t$ and $t+1$, then the ordering of the vertex robots around the convex hull is preserved from time $t$ to time $t+1$. This easily follows from the fact that every robot remains in its own movement region (cf.~Figure~\ref{fig2}).

\subsubsection{Convergence}\label{s:convergence}

We seek to prove that Algorithm~\ref{alg1} makes every robot eventually become a vertex robot. As it will be apparent in the proof of Theorem~\ref{t:correctness}, the crux of the problem is the situation in which only default moves are made (cf.~Figure~\ref{fig1:b}). We first prove that, if all robots perform only default moves, then they all converge to the same point (see Lemma~\ref{l:default} below).

Since we are assuming that only the vertex robots move, and that their movements depend only on the positions of other visible vertex robots, we may as well assume that all robots are vertex robots, and that their indices follow their order around the convex hull. Indeed, by the invariants observed in Section~\ref{s:invariants}, all robots will remain vertex robots throughout the execution, and their ordering around the convex hull will remain the same. So, let $r_{i-1}$, $r_i$, $r_{i+1}$ be three vertex robots, which appear on the boundary of $\mathcal{H}(t)$ consecutively in this order. Let $r_i$ perform a default move at time $t$. Then, the new position of $r_i$ is a convex combination of the current positions of these three robots, and precisely
\begin{equation}\label{e:default}
r_i(t+1)=\frac{r_{i-1}(t)}{4}+\frac{r_i(t)}{2}+\frac{r_{i+1}(t)}{4}.
\end{equation}
In general, as different sets of vertex robots are activated in several rounds, and nothing but default moves are made, the new location of each robot is always a convex combination of the \emph{original} positions of all the robots, obtained by applying~\eqref{e:default} to the set of active robots, at every round. In formulas,
$$r_i(t_0+t)=\sum_{j=1}^n \alpha_{i,j,t}\cdot r_j(t_0),$$
with $\alpha_{i,j,t}\geqslant 0$ and $\sum_{j=1}^n \alpha_{i,j,t}=1$, assuming that the robots start making only default moves at time $t_0$. Let $I=\{1,2,\cdots,n\}$. 
We fix $j\in I$, and we let $w_{i,t}=\alpha_{i,j,t}-\alpha_{i-1,j,t}$, where indices are taken modulo $n$. We claim that
\begin{equation}\label{e:claim}
\lim_{t\to\infty}w_{1,t}=\lim_{t\to\infty}w_{2,t}=\cdots=\lim_{t\to\infty}w_{n,t}=0.
\end{equation}
If such a claim is true (for all $j\in I$), it implies that the robots get arbitrarily close to each other, as $t$ grows. This, paired with the fact that $\mathcal H(t_0+t+1)\subseteq \mathcal H(t_0+t)$ for every $t$, as observed in Section~\ref{s:invariants}, allows us to conclude that the robots converge to the same limit point.

A proof of this statement can be obtained using the theory of convergence of asynchronous algorithms in the book~\cite{BeTsi89}. Indeed, the update rule~\eqref{e:default} corresponds to performing time stepping on a Markov chain with circulant transition matrix
\[
P=\begin{bmatrix}
1/2 & 1/4 & & & 1/4\\
1/4 & 1/2 & 1/4\\
& 1/4 & 1/2 & \ddots\\
& & \ddots & \ddots & 1/4\\
1/4& & & 1/4 & 1/2
\end{bmatrix}.
\]
It is proven in a statement on~\cite[page~435]{BeTsi89} that the time-stepping iteration converges even when performed asynchronously, under a model that generalizes our \SSynch.

Nevertheless, we give here an alternative self-contained proof. First we reformulate the problem in the following terms.

\paragraph{\underline{Communicating Vessels}} Suppose that $n$ vessels containing water are arranged in a circle, and there is a pipe between each pair of adjacent vessels, regulated by a valve. At every second, some of the valves are opened and others are closed, in such a way that each of the $n$ valves stays open for infinitely many seconds, in total. If a valve between two adjacent vessels stays open between seconds $t$ and $t+1$, then $1/4$ of the surplus of water, measured at second $t$, flows from the fuller vessel to the emptier one. Our claim is that the amount of water converges to the same limit in all vessels, no matter how the valves are opened and closed. We call this problem \Comm.
\vspace{10pt}

In this formulation, the amount of water in the $i$-th vessel at time $t\in\mathbb{N}$ would be our previous $w_{i,t}$. However, here we somewhat abstract from the \Obst problem, and we consider a slightly more general initial configuration, in which the $w_{i,0}$'s are arbitrary real numbers.

This problem is a special case of a diffusion model on a simple circular graph. To solve it, we shall introduce a quadratic energy functional $\norm{w_t}^2$, and prove that it is decreasing. The use of such an energy functional in this class of problems is well known in the literature (see for instance~\cite{Cyb89}), but the fact that the iteration is performed semi-synchronously on each node separately is less standard, so we need to do a little more work.

We set $v_{i,t}=1$ if the valve between the $i$-th and the $(i+1)$-th vessel is open between time $t$ and $t+1$ (indices are taken modulo $n$), and $v_{i,t}=0$ otherwise. It is easy to verify that activating robot $r_i$ at time $t$ in our previous discussion corresponds to setting $v_{i,t}=1$ in the \Comm formulation.

Let us denote by $w_t$ the vector whose $i$-th entry is $w_{i,t}$, and let $q_{i,t}=w_{i+1,t}-w_{i,t}$. We first prove an inequality on the Euclidean norms of the vectors $w_t$. Note that the inequality holds regardless of what assumptions are made on the opening pattern of the valves.

\begin{lemma}\label{l:main}
For every $t\in\mathbb{N}$,
\begin{equation}\label{e:main}
 \norm{w_t}^2-\norm{w_{t+1}}^2 \geqslant \frac{1}{4}\sum_{i=1}^n v_{i,t}\cdot q_{i,t}^2.
\end{equation}
\end{lemma}
\begin{proof}
For brevity, let $a=w_{i-1,t}$, $b=w_{i,t}$, $c=w_{i+1,t}$; hence, $q_{i-1,t}=b-a$ and $q_{i,t}=c-b$.

Suppose first that $v_{i-1,t}=v_{i,t}=1$, i.e., both valves connecting the $i$-th vessel with its neighbors are open. Then, $w_{i,t+1}=(a+2b+c)/4$. We have
\begin{equation}\label{case1}
\frac{w_{i-1,t}^2}{4}+\frac{w_{i,t}^2}{2}+\frac{w_{i+1,t}^2}{4}-w_{i,t+1}^2 \geqslant \frac{q_{i-1,t}^2}{8}+\frac{q_{i,t}^2}{8},
\end{equation}
which can be obtained by dropping the term $(a-c)^2/16$ from the algebraic identity
$$\frac{a^2}{4}+\frac{b^2}{2}+\frac{c^2}{4}-\frac{(a+2b+c)^2}{16}=\frac{(a-b)^2}{8}+\frac{(b-c)^2}{8}+\frac{(a-c)^2}{16}.$$
Now, suppose instead that $v_{i-1,t}=1$ and $v_{i,t}=0$. Then we have $w_{i,t+1}=(a+3b)/4$, and
\begin{equation} \label{case2}
\frac{w_{i-1,t}^2}{4}+\frac{3w_{i,t}^2}{4}-w_{i,t+1}^2 = \frac{3q_{i-1,t}^2}{16}\geqslant \frac{q_{i-1,t}^2}{8},
\end{equation}
where the first equality comes from the identity
$$\frac{a^2}{4}+\frac{3b^2}{4}-\frac{(a+3b)^2}{16} = \frac{3(a-b)^2}{16}.$$
If $v_{i-1,t}=0$ and $v_{i,t}=1$, an analogous argument gives
\begin{equation} \label{case3}
\frac{3w_{i,t}^2}{4}+\frac{w_{i+1,t}^2}{4}-w_{i,t+1}^2\geqslant \frac{q_{i,t}^2}{8}.
\end{equation}
Finally, if $v_{i-1,t}=v_{i,t}=0$, $w_{i,t+1}=w_{i,t}$, and trivially
\begin{equation} \label{case4}
w_{i,t}^2-w_{i,t+1}^2 = 0.
\end{equation}
We sum for each $i\in I$ the relevant inequality among \eqref{case1}, \eqref{case2}, \eqref{case3}, \eqref{case4}, depending on the value of $v_{i-1,t}$ and $v_{i,t}$.
Each of the terms $q_{i,t}^2/8$ appears twice if and only if $v_{i,t}=1$, and the coefficients of the terms in $w_{i,t}^2$ sum to 1 for every $i$, hence we get~\eqref{e:main}.
\end{proof}

From the previous lemma, it immediately follows that the sequence $(\norm{w_t})_{t\geqslant 0}$ is non-increasing. Since it is also bounded below by $0$, it converges to a limit, which we call $\ell$. Let $M_t=\max_{i\in I}\{w_{i,t}\}$ and $m_t=\min_{i\in I}\{w_{i,t}\}$. Observe that each entry of $w_{t+1}$ is a convex combination of entries of $w_t$, hence $(M_t)_{t\geqslant 0}$ is non-increasing and $(m_t)_{t\geqslant 0}$ is non-decreasing. Therefore they both converge, and we let $M=\lim_{t\to\infty} M_t$ and $m=\lim_{t\to\infty} m_t$.

\begin{corollary}\label{c:limits}
$$m \leqslant \frac{\ell}{\sqrt{n}}\leqslant M.$$
\end{corollary}
\begin{proof}
For every $t\in\mathbb{N}$, we have
$$nM_t^2 \geqslant \sum_{i=1}^n w_{i,t}^2 =\norm{w_t}^2\geqslant \ell^2,$$
which proves the second inequality. As for the first inequality, for every $\varepsilon>0$ and large-enough $t$, we have $nm_t^2\leqslant \norm{w_t}^2\leqslant \ell^2+\varepsilon$.
\end{proof}

For the next lemma, we let $V_i=\{t\in\mathbb{N}\mid v_{i,t}=1\}$.

\begin{lemma}\label{l:vessels}
Suppose that $\abs{V_i}=\infty$ for at least $n-1$ distinct values of $i\in I$. Then,
$$M=m=\frac{\ell}{\sqrt{n}}.$$
\end{lemma}
\begin{proof}
Due to Corollary~\ref{c:limits}, it is enough to prove that $M-m=0$. By contradiction, assume $M-m>0$, and let $\delta=(M-m)/(n+1)>0$. We have
$$\lim_{t\to\infty}\left(\norm{w_t}^2-\norm{w_{t+1}}^2\right)=\ell^2-\ell^2=0,$$
hence there exists $T\in\mathbb{N}$ such that $\norm{w_t}^2-\norm{w_{t+1}}^2 < \delta^2/4$ for every $t\geqslant T$. By Lemma~\ref{l:main},
$$\frac{q_{i,t}^2}{4}\leqslant\norm{w_t}^2-\norm{w_{t+1}}^2<\frac{\delta^2}{4}$$
for every $t\geqslant T$ and every $i$ such that $v_{i,t}=1$. This implies $\abs{q_{i,t}}<\delta$, that is, a necessary condition for the valve between the $i$-th and the $(i+1)$-th vessel to be open at time $t\geqslant T$ is that $\abs{w_{i+1,t}-w_{i,t}}<\delta$. Consider now the $n+1$ open intervals
$$(m,m+\delta), (m+\delta,m+2\delta), \cdots, (m+n\delta,M),$$
each of width $\delta$. Since $M_T\geqslant M$ and $m_T\leqslant m$, there are $w_{i,T}$'s above and below all these intervals. Moreover, by the pigeonhole principle, at least one of the intervals contains no $w_{i,T}$'s, for any $i\in I$. In other words, we can find a partition $I_1 \cup I_2=I$, with $I_1$ and $I_2$ both non-empty, and a threshold value $\lambda$ such that $w_{i,T}\leqslant\lambda$ for every $i\in I_1$, and $w_{i,T}\geqslant\lambda+\delta$ for every $i\in I_2$. Hence, at time $T$, only valves between entries of $w_t$ whose indices belong to the same $I_k$ can be open. It is now easy to prove by induction on $t\geqslant T$ the following facts:
\begin{itemize}
\item $\max_{i\in I_1}\{w_{i,t}\}\leqslant\lambda$,
\item $\min_{i\in I_2}\{w_{i,t}\}\geqslant\lambda+\delta$,
\item $v_{i,t}=0$ whenever $i$ and $i+1$ belong to two different classes of the partition.
\end{itemize}
Since $I_1$ and $I_2$ are non-empty, there must be at least two distinct indices $i'\in I_1$ and $i''\in I_2$ such that $i'+1\in I_2$ and $i''+1\in I_1$ (where indices are taken modulo $n$). It follows that the $i'$-th and $i''$-th valve are never open for $t\geqslant T$, and this contradicts the hypothesis that $\abs{V_i}<\infty$ for at most one choice of $i\in I$.
\end{proof}

This solves the \Comm problem.

\begin{corollary}\label{c:vessels}
Under the hypotheses of Lemma~\ref{l:vessels}, for every $i\in I$,
$$\lim_{t\to\infty} w_{i,t}=\frac{\ell}{\sqrt{n}}=\frac{\sum_{j=1}^n w_{j,0}}{n}.$$
\end{corollary}
\begin{proof}
By Lemma~\ref{l:vessels}, since $m_t\leqslant w_{i,t}\leqslant M_t$, all the limits coincide. Moreover, the sum of the $w_{i,t}$'s does not depend on $t$; hence their average, taken at any time, must be equal to the joint limit.
\end{proof}

Let us return to the \Obst problem, to prove our final lemma.

\begin{lemma}\label{l:default}
If, at every round, each robot makes a default move (cf.~Figure~\ref{fig1:b}) or stays still, all external robots have their lights set to \emph{Vertex}, and no new robots become vertex robots or terminate, then all robots' locations converge to the same limit point.
\end{lemma}
\begin{proof}
As discussed at the beginning of Section~\ref{s:convergence}, this is implied by~\eqref{e:claim}. Recall that $w_{i,0}=\alpha_{i,j,0}-\alpha_{i-1,j,0}$, and hence $\sum_{i=1}^n w_{i,0}=0$. Then,~\eqref{e:claim} follows immediately from Corollary~\ref{c:vessels}.
\end{proof}

We are now ready to prove our main theorem.

\begin{theorem}\label{t:correctness}
Algorithm~\ref{alg1} solves \Obst for \Rigid\ \SSynch robots with 2-colored lights.
\end{theorem}
\begin{proof}
If the initial convex hull is a line segment, it becomes a non-degenerate polygon as soon as one or more of the non-vertex robots are activated. It is also easy to observe (cf.~Figure~\ref{fig2}) that, from this configuration, the convex hull may never become a line segment. So the invariants discussed in Section~\ref{s:invariants} apply, possibly after a few initial rounds: no two robots will ever collide, and a vertex robot will never become a non-vertex robot.

Assume by contradiction that the execution never terminates. Note that a robot terminates if and only if all robots terminate. Indeed, if there are any non-vertex robots (whose lights are still set to \emph{Off}), then each vertex robot can see at least one of them. Hence we are assuming that all robots execute the algorithm forever.

At some point, the set of vertex robots reaches a maximum $\mathcal M\subseteq \mathcal R$, and as soon as all of these robots have been activated, they permanently set their lights to \emph{Vertex}. Let $T\in\mathbb N$ be a time at which all the robots in $\mathcal M$ have their lights set to \emph{Vertex}. Suppose that there are external robots that are not vertex robots after time $T$, and let $r$ be one such robot that is adjacent to a vertex robot $r'$. Then, after $r'$ is activated and moves, $r$ becomes a vertex robot as well, contradicting the maximality of $\mathcal M$. Hence the external robots are exactly the robots in $\mathcal M$, and no other robot may become external after time $T$.

If there is only one internal robot at time $t\geqslant T$, it becomes external as soon as it is activated, due to line~23 of the algorithm, which is impossible, as argued in the previous paragraph. Therefore there are at least two internal robots at every time $t\geqslant T$. On the other hand, if a vertex robot makes a non-default move at any time $t\geqslant T$, a new robot becomes external at time $t+1$. Indeed, referring to Figure~\ref{fig1:a}, the line $uv$ passes through $p(t+1)$ and $c(t+1)$, and no robot lies above this line at time $t+1$. Hence $c$ becomes a new external robot, which again is impossible.

As a consequence, only default moves are made after time $T$. Moreover, no robot becomes external or becomes a vertex robot after time $T$, and no robot ever terminates. Therefore Lemma~\ref{l:default} applies, and the robots converge to the same limit point. But since there are at least two internal robots, this means that at least one of them has to move, implying that it becomes a vertex robot at some point (by the above assumptions, only vertex robots can move), a contradiction.

Hence the execution terminates, meaning that at some point one of the robots sees only vertex robots. This implies that there are no non-vertex robots, hence the configuration is strictly convex, all robots can see each other, and they all terminate without moving as soon as they are activated, thus solving the \Obst problem.
\end{proof}

\section{Solving \Obst for \NRigid\ \SSynch Robots}\label{s:ssynch2}

Here we give a protocol, Algorithm~\ref{alg2} (\emph{Contain}), for the \Obst problem that works for \NRigid robots and the \SSynch scheduler. Recall that, in the \NRigid model, the robots make unreliable moves, that is, the scheduler can stop them before they reach their destination point, but not before they have moved by at least a constant $\delta>0$. Since these robots are weaker than the ones considered in Section~\ref{s:ssynch}, they will require lights of three possible colors, as opposed to two.

Our goal is also to design an algorithm that can be applied to robots in the \Rigid\ \ASynch model, as well as the \NRigid\ \SSynch one. This model will be discussed in Section~\ref{s:rasynch}. In order to do this, we introduce a couple of extra technical subtleties into Algorithm~\ref{alg2}, which are irrelevant here, but will turn out to be necessary in Section~\ref{s:rasynch}.

\subsection{Description of Algorithm~\ref{alg2}}\label{s:alg2descr}

Algorithm~\ref{alg2} consists of three phases, to be executed in succession: a \emph{segment breaking} phase, an \emph{interior depletion} phase, and a \emph{vertex adjustments} phase. The first phase deals with the special configuration in which the robots are all collinear, and makes them not collinear. If the robots are not initially collinear, this phase is skipped. In the second phase, the internal robots move toward the boundary of the convex hull, thus forming a convex configuration, perhaps with some degenerate vertices. In the third phase the robots (which are now all external) make small movements to finally reach a strictly convex configuration. Three colors are used by the robots:  $\mathcal C=\{\mbox{\emph{Off}},\mbox{\emph{External}},\mbox{\emph{Adjusting}}\}$. Initially, all robots' lights are set to \emph{Off}.

For added clarity, in the algorithm the line numbers of instructions belonging logically to different phases are typeset in different colors, according to the following table.

\newcommand{\segmentcolor}[1]{\textcolor{red}{#1}}
\newcommand{\interiorcolor}[1]{\textcolor{ForestGreen}{#1}}
\newcommand{\vertexcolor}[1]{\textcolor{blue}{#1}}
\newcommand{\allcolor}[1]{\textcolor{black}{#1}}
\begin{center}
 \begin{tabular}{rl}
  \toprule
  Segment breaking & \segmentcolor{red}\\
  Interior depletion & \interiorcolor{green}\\
  Vertex adjustments & \vertexcolor{blue}\\
  \bottomrule
 \end{tabular}
\end{center}

\SetAlgoCaptionSeparator{: Contain.}
\begin{algorithm}\label{alg2}
\caption{Solving the \Obst problem for \NRigid\ \SSynch robots and \Rigid\ \ASynch robots with 3-colored lights}
\DontPrintSemicolon
\KwIn{$\mathcal V$: set of robots visible to me (myself included) whose positions are expressed in a coordinate system centered at my location.}
\nlset{\segmentcolor{1}\ \interiorcolor{1}\ \vertexcolor{1}}$r^* \longleftarrow$ myself\;
\nlset{\segmentcolor{2}\ \interiorcolor{2}\ \vertexcolor{2}}$\mathcal P \longleftarrow \{r.\mbox{\emph{position}}\mid r\in \mathcal V\}$\;
\nlset{\segmentcolor{3}\ \interiorcolor{3}\ \vertexcolor{3}}$\mathcal H \longleftarrow$ convex hull of $\mathcal P$\;
\nlset{\segmentcolor{4}\ \interiorcolor{4}\ \vertexcolor{4}}$\partial \mathcal H\longleftarrow$ boundary of $\mathcal H$\;
\nlset{\segmentcolor{5}}\lIf{$|\mathcal V|=1$}{Terminate\;}
\nlset{\segmentcolor{6}}\ElseIf{$|\mathcal V|=2$}{
\nlset{\segmentcolor{7}}	\If{$r^*.\mbox{light}=\mbox{Adjusting}$}{
\nlset{\segmentcolor{8}}		$r^*.\mbox{\emph{light}} \longleftarrow \mbox{\emph{External}}$\;
\nlset{\segmentcolor{9}}		Terminate\;
	}
\nlset{\segmentcolor{10}}	\Else{
\nlset{\segmentcolor{11}}		$r^*.\mbox{\emph{light}} \longleftarrow \mbox{\emph{Adjusting}}$\;
\nlset{\segmentcolor{12}}		Move orthogonally to $\mathcal H$ by the length of $\mathcal H$\;
	}
}
\nlset{\segmentcolor{13}}\ElseIf{$\mathcal H$ is a line segment}{
\nlset{\segmentcolor{14}}	\If{$\forall r\in \mathcal V\setminus\{r^*\},\, r.\mbox{light}=\mbox{External}$}{
\nlset{\segmentcolor{15}}		$r^*.\mbox{\emph{light}} \longleftarrow \mbox{\emph{Adjusting}}$\;
\nlset{\segmentcolor{16}}		Move orthogonally to $\mathcal H$ by any positive amount\;
	}
}
\end{algorithm}

\setcounter{algocf}{1}
\SetAlgoCaptionSeparator{: Contain (continued).}
\begin{algorithm}
\caption{}
\DontPrintSemicolon
\nlset{\segmentcolor{17}\ \vertexcolor{17}}\ElseIf{$r^*.\mbox{position}\in\partial\mathcal H$}{
\nlset{\segmentcolor{18}\ \vertexcolor{18}}	$a \longleftarrow$ my ccw-neighboring robot on $\partial\mathcal H$\;
\nlset{\segmentcolor{19}\ \vertexcolor{19}}	$b \longleftarrow$ my cw-neighboring robot on $\partial\mathcal H$\;
\nlset{\segmentcolor{20}\ \vertexcolor{20}}\If{$r^*.\mbox{light}=\mbox{Adjusting}$}{
\nlset{\segmentcolor{21}\ \vertexcolor{21}}				\If{$\forall r\in\mathcal V,\, r.\mbox{light}\neq\mbox{Off}$\\
\nlset{\segmentcolor{22}\ \vertexcolor{22}}				\mbox{\textbf{\emph{or}} }$\exists r\in\mathcal V,\, r.\mbox{light}=\mbox{External}$}{
\nlset{\segmentcolor{23}\ \vertexcolor{23}}		$r^*.\mbox{\emph{light}} \longleftarrow \mbox{\emph{External}}$\;
\nlset{\segmentcolor{24}\ \vertexcolor{24}}		\lIf{$a.\mbox{light}\neq\mbox{Off}$\mbox{ \textbf{\emph{and}} }$b.\mbox{light}\neq\mbox{Off}$\\
\nlset{\segmentcolor{25}\ \vertexcolor{25}}\mbox{\textbf{\emph{and}} }$(\mathcal H\setminus \partial\mathcal H)\cap \mathcal P=\varnothing$}{Terminate}
				}
	}
\nlset{\vertexcolor{26}}\ElseIf{$r^*.\mbox{position}$ is a non-degenerate vertex of $\mathcal H$\\
\nlset{\vertexcolor{27}}\mbox{\textbf{\emph{and}} }$\forall r\in\mathcal V,\, r.\mbox{light}=\mbox{External}$}{
\nlset{\vertexcolor{28}}	$r^*.\mbox{\emph{light}}\longleftarrow \mbox{\emph{Adjusting}}$\;
\nlset{\vertexcolor{29}}	Move to $(a.\mbox{\emph{position}}+b.\mbox{\emph{position}})/4$\;
}
\nlset{\segmentcolor{30}\ \interiorcolor{30}\ \vertexcolor{30}}\Else{
\nlset{\segmentcolor{31}\ \interiorcolor{31}\ \vertexcolor{31}}$\mathcal W\longleftarrow \{r\in \mathcal V \mid r.\mbox{\emph{light}}=\mbox{\emph{Adjusting}}\}$\;
\nlset{\segmentcolor{32}\ \interiorcolor{32}\ \vertexcolor{32}}\If{\emph{\textbf{(}}$|\mathcal V|=3$\emph{\textbf{ and }}the internal angle of $\mathcal H$ at $r^*.\mbox{position}$ is acute\emph{\textbf{)}}\\
\nlset{\segmentcolor{33}\ \interiorcolor{33}\ \vertexcolor{33}}\emph{\textbf{or (}}$|\mathcal W|>1$\emph{\textbf{ and }}$r^*.\mbox{position}$ is a non-degenerate vertex of $\mathcal H$\emph{\textbf{)}}\\
\nlset{\segmentcolor{34}\ \interiorcolor{34}\ \vertexcolor{34}}\emph{\textbf{or }}$\mathcal W=\varnothing$}{
\nlset{\segmentcolor{35}\ \interiorcolor{35}\ \vertexcolor{35}}$r^*.\mbox{\emph{light}} \longleftarrow \mbox{\emph{External}}$
	}
}
}
\end{algorithm}

\setcounter{algocf}{1}
\SetAlgoCaptionSeparator{: Contain (continued).}
\begin{algorithm}
\caption{}
\DontPrintSemicolon
\nlset{\interiorcolor{36}}\ElseIf{$\forall r\in\mathcal V,\, r.\mbox{light}\neq\mbox{Adjusting}$}{
\nlset{\interiorcolor{37}}	$\mathcal P' \longleftarrow \{r.\mbox{\emph{position}}\mid r\in \mathcal V \wedge r.\mbox{\emph{light}}=\mbox{\emph{Off}}\}$\;
\nlset{\interiorcolor{38}}	$\mathcal H' \longleftarrow$ convex hull of $\mathcal P'$\;
\nlset{\interiorcolor{39}}	$\partial \mathcal H'\longleftarrow$ boundary of $\mathcal H'$\;
\nlset{\interiorcolor{40}}	\If{$|\mathcal P'|=1$}{
\nlset{\interiorcolor{41}}Move to a closest midpoint of a connected component of $\partial\mathcal H\setminus \mathcal P$}
\nlset{\interiorcolor{42}}	\ElseIf{$|\mathcal P'|=2$}{
\nlset{\interiorcolor{43}}		$\ell\longleftarrow$ line containing $\mathcal H'$\;
\nlset{\interiorcolor{44}}		$\mathcal A\longleftarrow$ right angle with axis of symmetry $\ell$ such that\\
$\qquad\quad\, \mathcal A\cap \mathcal H'=\{r^*.\mbox{\emph{position}}\}$\;
\nlset{\interiorcolor{45}}		Move to any point of $(\mathcal A\cap\partial\mathcal H)\setminus \mathcal P$\;
	}
\nlset{\interiorcolor{46}}	\ElseIf{$r^*.\mbox{position}$ is a non-degenerate vertex of $\mathcal H'$}{
\nlset{\interiorcolor{47}}		$\mathcal A\longleftarrow$ internal angle of $\mathcal H'$ whose vertex is $r^*.\mbox{\emph{position}}$\;
\nlset{\interiorcolor{48}}		$\alpha\longleftarrow$ measure of $\mathcal A$\;
\nlset{\interiorcolor{49}}		$\ell\longleftarrow$ axis of symmetry of $\mathcal A$\;
\nlset{\interiorcolor{50}}		\lIf{$\alpha\leqslant \pi/2$}{$\alpha'\longleftarrow\alpha$\;}
\nlset{\interiorcolor{51}}		\lElse{$\alpha'\longleftarrow \pi-\alpha$\;}
\nlset{\interiorcolor{52}}		$\mathcal A'\longleftarrow$ angle of measure $\alpha'$ with axis of symmetry $\ell$ such that\\
$\qquad\ \ \ \ \mathcal A'\cap \mathcal H'=\{r^*.\mbox{\emph{position}}\}$\;
\nlset{\interiorcolor{53}}		$\mathcal E\longleftarrow \{p\in \partial\mathcal H\mid \exists a,b\in\mathcal P \setminus \mathcal A,\, p\in ab\}$\;
\nlset{\interiorcolor{54}}		\lIf{$(\mathcal A'\cap\mathcal E)\setminus \mathcal P\neq\varnothing$}{Move to any point of $(\mathcal A'\cap\mathcal E)\setminus \mathcal P$}
	}
}
\end{algorithm}
\SetAlgoCaptionSeparator{:}

Recall that we denote by $\mathcal H(t)$ the convex hull of the positions of all the robots at time $t\in\mathbb N$. In this section we also denote by $\mathcal H'(t)$ the convex hull of the positions of the internal robots at time $t\in\mathbb N$. Note that the ``global'' notions of $\mathcal H$ and $\mathcal H'$ may differ from the ones computed by the robots executing Algorithm~\ref{alg2}, because a robot may be unable to see the positions of all the other robots in the swarm. In the following discussion, when referring to $\mathcal H$ and $\mathcal H'$, we will typically mean the ``global'' ones, unless we explicitly state otherwise.

We first describe the interior depletion phase, starting from a non-collinear initial configuration. To begin with, all the robots' lights are set to \emph{Off}. As soon as an external robot is activated, it sets its own light to \emph{External} (lines~34, 35) and does not move as long as it can still see robots whose light is \emph{Off} (lines~26, 27). Note that a robot $r$ that occupies a vertex of $\mathcal H'$ eventually becomes aware of it, by looking at the convex hull of the visible robots whose lights are \emph{Off}. These may not all be internal robots, because perhaps not all external robots have been activated yet, but eventually $r$ gets to see a good-enough approximation of a ``neighborhood'' of $\mathcal H'$, and it realizes it occupies one of its vertices.

So, when a robot understands that it lies on a vertex of $\mathcal H'$, it moves toward the boundary of $\mathcal H$, part of which is also identifiable by $r$. We distinguish three cases.
\begin{enumerate}
\item If $r$ realizes it is the only internal robot, it moves toward the midpoint of an edge of the convex hull (line~41). To avoid bouncing back and forth at different turns, it always chooses the closest of such midpoints.
\item If $r$ realizes that $\mathcal H'$ is a line segment and it occupies one endpoint of it, it moves like in Figure~\ref{fig3:a}. That is, it moves to the boundary of $\mathcal H$, while remaining within a right angle oriented away from $\mathcal H'$ (lines~43--45).

\begin{figure}[h]
\centering
\subfigure[Case with collinear internal robots]{\label{fig3:a}\includegraphics[scale=.75]{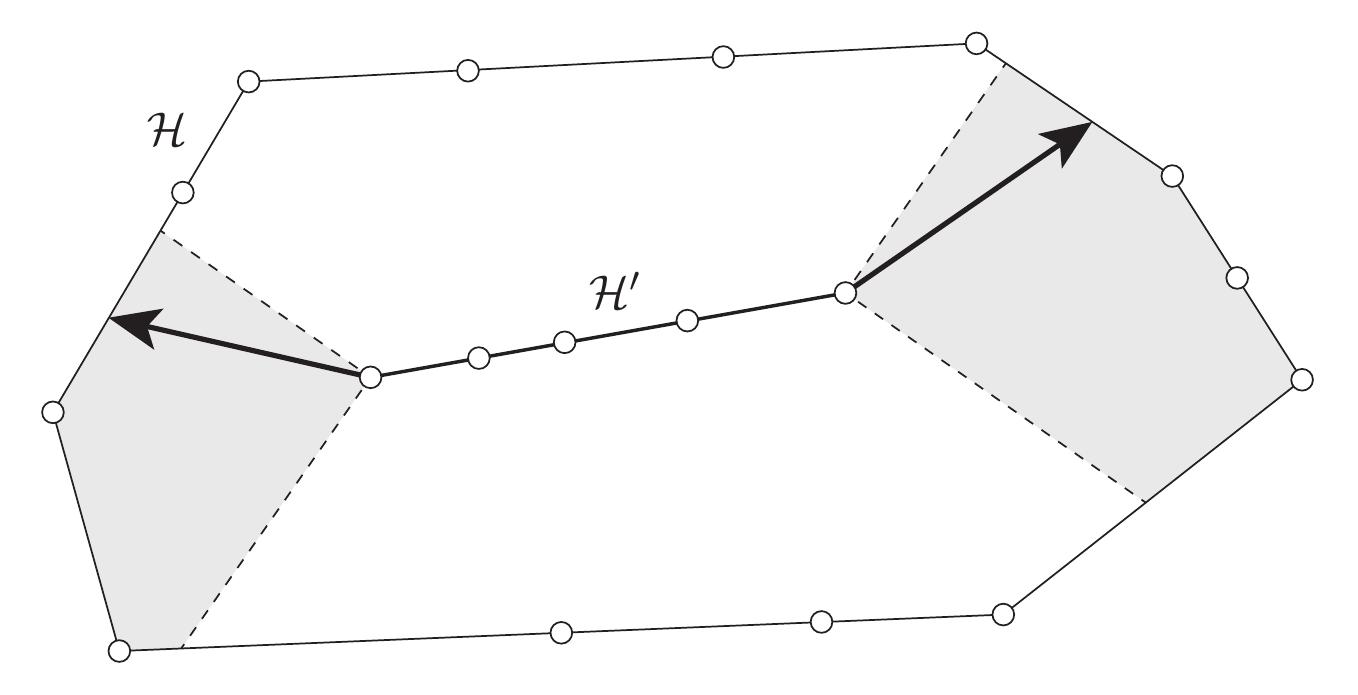}}
\subfigure[General case]{\label{fig3:b}\includegraphics[scale=.75]{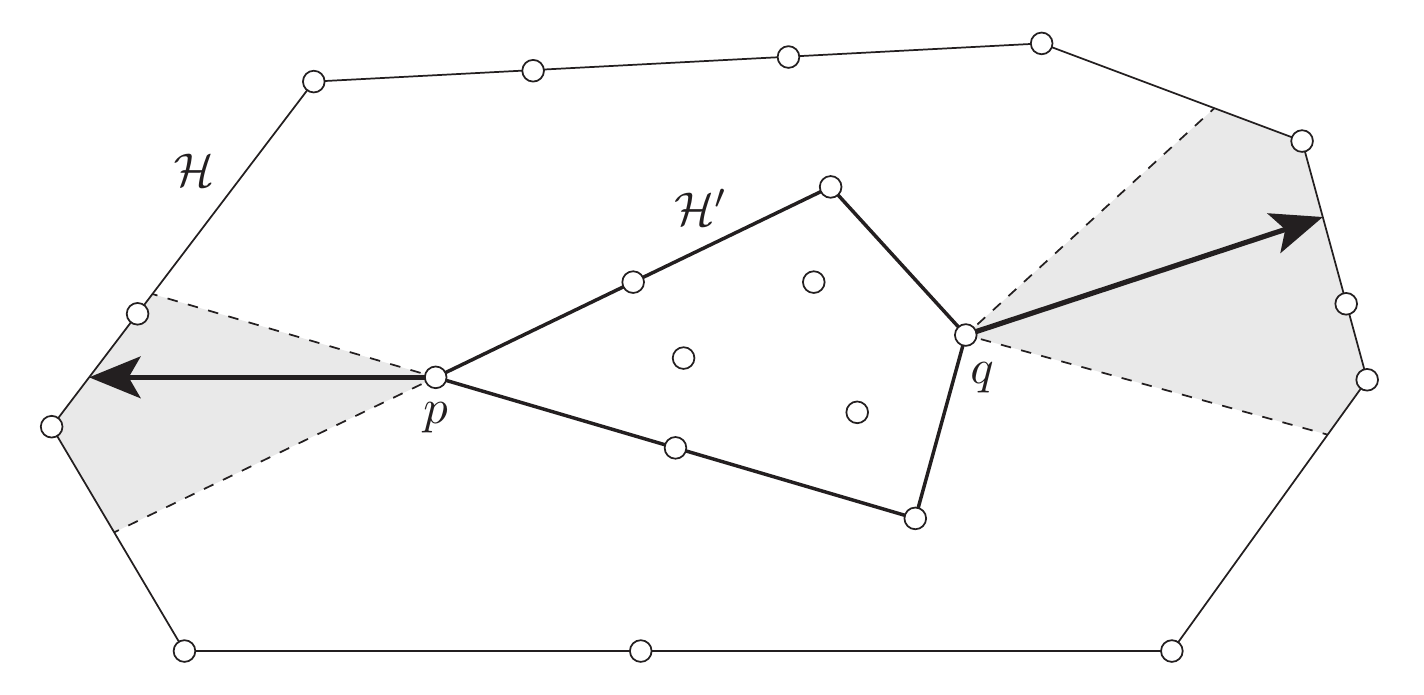}}
\caption{Interior depletion phase}
\label{fig3}
\end{figure}

\item Finally, if $r$ ``believes'' that $\mathcal H'$ is a non-degenerate polygon and that it lies on one of its vertices, it moves as in Figure~\ref{fig3:b} (lines~47--54). Remember that $r$ may believe so even if $\mathcal H'$ is actually degenerate, because some external robots may still be \emph{Off}. However, $r$ gets an approximation of $\mathcal H'$, which we call $\mathcal H'_r$, and it knows it lies on a vertex of $\mathcal H'_r$, implying that it also lies on a vertex of the ``real'' $\mathcal H'$. Now, if the internal angle of $\mathcal H'_r$ at $r(t)$ is acute, $r$ moves as the robot in $p$ in Figure~\ref{fig3:b}: it moves to the boundary of $\mathcal H$ while remaining between the extensions of its two incident edges of $\mathcal H'_r$. Otherwise, if the angle is not acute, $r$ moves as the robot in $q$ in Figure~\ref{fig3:b}: it moves to the boundary of $\mathcal H$ while staying between the two perpendiculars to its incident edges of $\mathcal H'_r$. Moreover, $r$ actually performs the move only if it is sure that its destination point lies on the boundary of the ``real'' $\mathcal H$. For this reason, it has to check if the destination point computed as described above lies on a completely-visible edge of the observed convex hull whose endpoints are both set to \emph{External} (line~53). For instance, in Figure~\ref{fig13}, the robot in $p$ cannot move to the gray area even if the robots in $a$ and $b$ are set to \emph{External}, because the robot in $q$ prevents the one in $p$ from seeing the whole edge $ab$. On the other hand, the robot in $q$ can move to its own gray area, provided that $a$ and $b$ are set to \emph{External}. Indeed, the robot in $q$ can see all of $ab$, and it is therefore sure that it is an edge of the ``real'' convex hull.
\end{enumerate}

\begin{figure}[h]
\centering
\includegraphics[scale=.9]{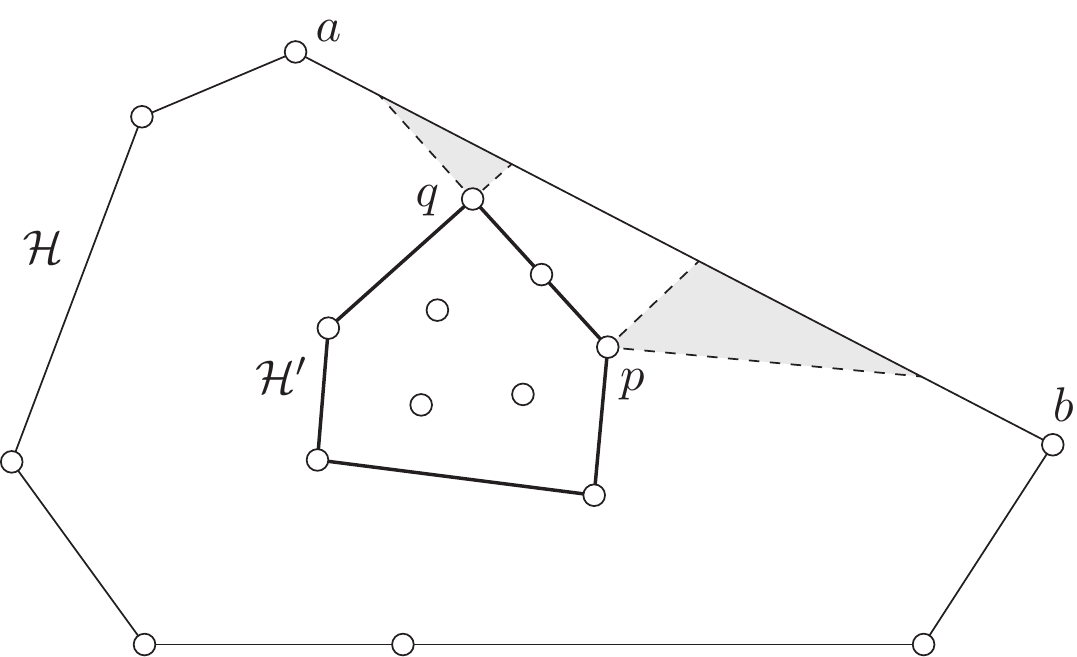}
\caption{The robot in $p$ cannot move even if $a$ and $b$ are set to \emph{External}, because $q$ may be hiding some other external robots, and $ab$ may not be an edge of the convex hull}
\label{fig13}
\end{figure}

Now to the vertex adjustments phase. When a robot lies at a vertex of $\mathcal H$ and it sees only robots whose light is set to \emph{External}, it makes the ``default move'' of Figure~\ref{fig1:b}, where $a$ and $b$ are the locations of its two neighbors on $\mathcal H$ (line~29). Moreover, while doing so it also sets its light to the third value, \emph{Adjusting}, as a ``self-reminder'' (line~28). So, when it is activated again, it knows it has already adjusted its position, and it terminates, after reverting its light to \emph{External} (lines~20--25). This way we make sure that each vertex robot adjusts its position exactly once, and we ensure termination. When the adjustment is done, the robots at $a$ and $b$ are guaranteed to occupy vertices of $\mathcal H$, instead of lying in the middle of an edge. So, each external robot becomes a vertex robot at some point, then it adjusts its position while remaining a vertex, possibly making its adjacent robots become vertices as well, and it terminates. When all robots have terminated, the configuration is strictly convex, and therefore \Obst is solved.

Finally, the segment breaking phase deals with the special case in which all robots are initially collinear. Let robots $r$ and $s$ be the two endpoints of $\mathcal H$: as soon as one of them is activated (possibly both), it sets its light to \emph{Adjusting}, moves orthogonally to $\mathcal H$, and then waits (lines~11, 12). Meanwhile, the other robots do not do anything until some conditions are met (lines~31--34). If only $r$ moves, $s$ realizes it (line~32) and sets its own light to \emph{External} (and vice versa). If both $r$ and $s$ move together, some other robot realizes that it is a non-degenerate vertex of the convex hull and that it can see both $r$ and $s$ set to \emph{Adjusting} (line~33): in this case, it sets itself to \emph{External}. When $r$ or $s$ sees some robots set to \emph{External}, it finally sets itself to \emph{External}, as well (lines~22, 23). Additionally, it may terminate, provided that neither of its neighboring robots on the convex hull's boundary has still its light set to \emph{Off} (line~24) and that it recognizes no robots as internal (line~25). This is to force $r$ and $s$ to make at least one default move in the unfortunate case that a third external robot is found between them after their initial move, or gets there during the interior depletion phase (refer to the complete discussion in Section~\ref{s:segbreak}). After this is done, the execution transitions seamlessly into one of the general cases.

If $n\leqslant 3$ this is not sufficient. Suppose first that $n=3$. Then, $r$ and $s$ may move in such a way that the configuration remains centrally symmetric, with the middle robot $q$ obstructing $r$ and $s$. However, after moving once, $r$ and $s$ become \emph{External} and terminate (lines~8, 9). Meanwhile $q$ waits until it sees both $r$ and $s$ set to \emph{External}, and finally it moves orthogonally to $\mathcal H$ (lines~13--16), thus solving \Obst also in this special case.

If $n=2$, each robot moves once (lines~11, 12), and then it detects a situation in which it can safely terminate (lines~6--9).

\subsection{Correctness of Algorithm~\ref{alg2}}

\subsubsection{Interior Depletion Phase}

We first prove that no collisions occur during the interior depletion phase, and then that the phase itself eventually terminates, with all the robots becoming external. In this section we will assume that the robots are not initially collinear. The collinear case will be discussed in Section~\ref{s:segbreak}, and it will be shown that is seamlessly transitions into one of the other cases.

It is easy to observe that, during the interior depletion phase, all external robots keep seeing (internal) robots whose lights are set to \emph{Off}, and therefore none of them moves. On the other hand, no internal robot moves outside of the convex hull.

\begin{observation}\label{obs:depletion}
If there are internal robots at time $t$, no external robot moves, and $\mathcal H(t)=\mathcal H(t+1)$.
\end{observation}

\begin{lemma}\label{l:coll1}
If $r$ and $s$ are two internal robots at time $t$, then
$$(r(t+1)-r(t))\bullet(s(t)-r(t))\leqslant 0.$$
\end{lemma}
\begin{proof}
If $r$ is not activated at time $t$, or it is activated but it does not move, then the left-hand side is zero, and therefore the inequality holds. Suppose now that $r$ moves by a positive amount, so $r(t+1)-r(t)$ is not the null vector. Let $\ell$ be the line through $r(t)$ that is orthogonal to the segment $r(t)r(t+1)$. By construction, $r$ moves in such a way that $r(t+1)$ lies in the open half-plane bounded by $\ell$ that does not contain $\mathcal H'(t)$ (note that this holds \emph{a fortiori} also if some external robots have not set their lights to \emph{External} yet, and therefore the $\mathcal H'$ computed by $r$ is larger than the real one). Since $s(t)\in \mathcal H'(t)$, $s(t)$ lies on $\ell$ or in the half-plane bounded by $\ell$ that does not contain $r(t+1)$. This is equivalent to saying that the dot product between $r(t+1)-r(t)$ and $s(t)-r(t)$ is not positive.
\end{proof}

\begin{lemma}\label{l:coll2}
As long as there are internal robots, no collisions occur.
\end{lemma}
\begin{proof}
If there are internal robots, every external robot sees robots whose light is set to \emph{Off}, and hence it does not move. By construction, the internal robots avoid moving on top of external robots, and therefore there can be no collision involving external robots.

Suppose by contradiction that two robots $r$ and $s$ that are internal at time $t$ collide for the first time at $t+1$, and therefore $r(t+1)=s(t+1)=p$. By Lemma~\ref{l:coll1} applied to $r$ and $s$, we have
\begin{equation}\label{eq:coll1}
(p-r(t))\bullet(s(t)-r(t))\leqslant 0.
\end{equation}
Applying Lemma~\ref{l:coll1} again with $r$ and $s$ inverted, we also have
\begin{equation}\label{eq:coll2}
(p-s(t))\bullet(r(t)-s(t))\leqslant 0.
\end{equation}
Adding~\ref{eq:coll1} and~\ref{eq:coll2} together and doing some algebraic manipulations, we obtain
$$(p-r(t))\bullet(s(t)-r(t))+(p-s(t))\bullet(r(t)-s(t))\leqslant 0,$$
$$(s(t)-r(t))\bullet((p-r(t))-(p-s(t)))\leqslant 0,$$
$$(s(t)-r(t))\bullet(s(t)-r(t))\leqslant 0.$$
The latter is equivalent to $\norm{s(t)-r(t)}\leqslant 0$, implying that $r(t)=s(t)$. This contradicts the fact that $r$ and $s$ collide for the first time at $t+1$.
\end{proof}

We still have to prove that the interior depletion phase terminates, that is, eventually all robots become external. Due to Observation~\ref{obs:depletion}, when a robot becomes external, it stops moving and remains external, at least as long as there are other internal robots. Thus, if by contradiction this phase does not terminate, the set of internal robots reaches a non-empty minimum, and from that time on no new robot becomes external. After possibly some more turns, say at time $T\in \mathbb N$, all external robots have been activated and have set their lights to \emph{External}, and hence no robot changes its light any more.

In the following lemmas, we will show that these assumptions on $T$ yield a contradiction. We will prove that, if $\mathcal H'(T)$ is a non-degenerate polygon, then either its area or its diameter will grow unboundedly. Therefore, at some point in time, $\mathcal H'$ will not be a subset of $\mathcal H$ any more. (The analysis when $\mathcal H'(T)$ is a degenerate polygon is easy, and it will be carried out in the proof of Lemma~\ref{l:intfinal}.)

Recall that, due to line~50 of Algorithm~\ref{alg2}, when a robot computes its destination, it remains within the extensions of its incident edges of $\mathcal H'(t)$. Hence, referring to Figure~\ref{fig4}, it is easy to observe the following.

\begin{observation}\label{o:edges}
Let robots $r$ and $s$ lie at adjacent vertices of $\mathcal H'(t)$ at time $t\geqslant T$, and let the area of $\mathcal H'(t)$ be positive. Then, $r(t+1)$ and $s(t+1)$ lie on the same side of the line through $r(t)$ and $s(t)$ (or possibly on the line itself).
\end{observation}

\begin{lemma}\label{l:intgrow}
If $t\geqslant T$ and the area of $\mathcal H'(t)$ is positive, then $\mathcal H'(t)\subseteq\mathcal H'(t+1)$.
\end{lemma}
\begin{proof}
Let $\mathcal R'\subset \mathcal R$ be the set of robots that lie at vertices of $\mathcal H'(t)$, at time $t\geqslant T$. Let $\mathcal P$ be the polygon (illustrated in Figure~\ref{fig4} as a thick dashed polygon) whose vertices are the locations at time $t+1$ of the robots of $\mathcal R'$, taken in the same order as they appear around the boundary of $\mathcal H'$. Note that, since the robots are \NRigid and not all of them are necessarily activated at time $t$, $\mathcal P$ is not necessarily a convex polygon.

Because the property stated in Observation~\ref{o:edges} holds for all the edges of $\mathcal H'(t)$ and $\mathcal P$, we have that $\mathcal H'(t)\subseteq \mathcal P$. But, by definition of $T$, none of the robots of $\mathcal R'$ ever becomes external, and hence $\mathcal H'(t+1)$ is the convex hull of $\mathcal P$. We conclude that $\mathcal H'(t)\subseteq \mathcal P\subseteq \mathcal H'(t+1)$.
\end{proof}

\begin{figure}[h]
\centering
\includegraphics[width=.95\linewidth]{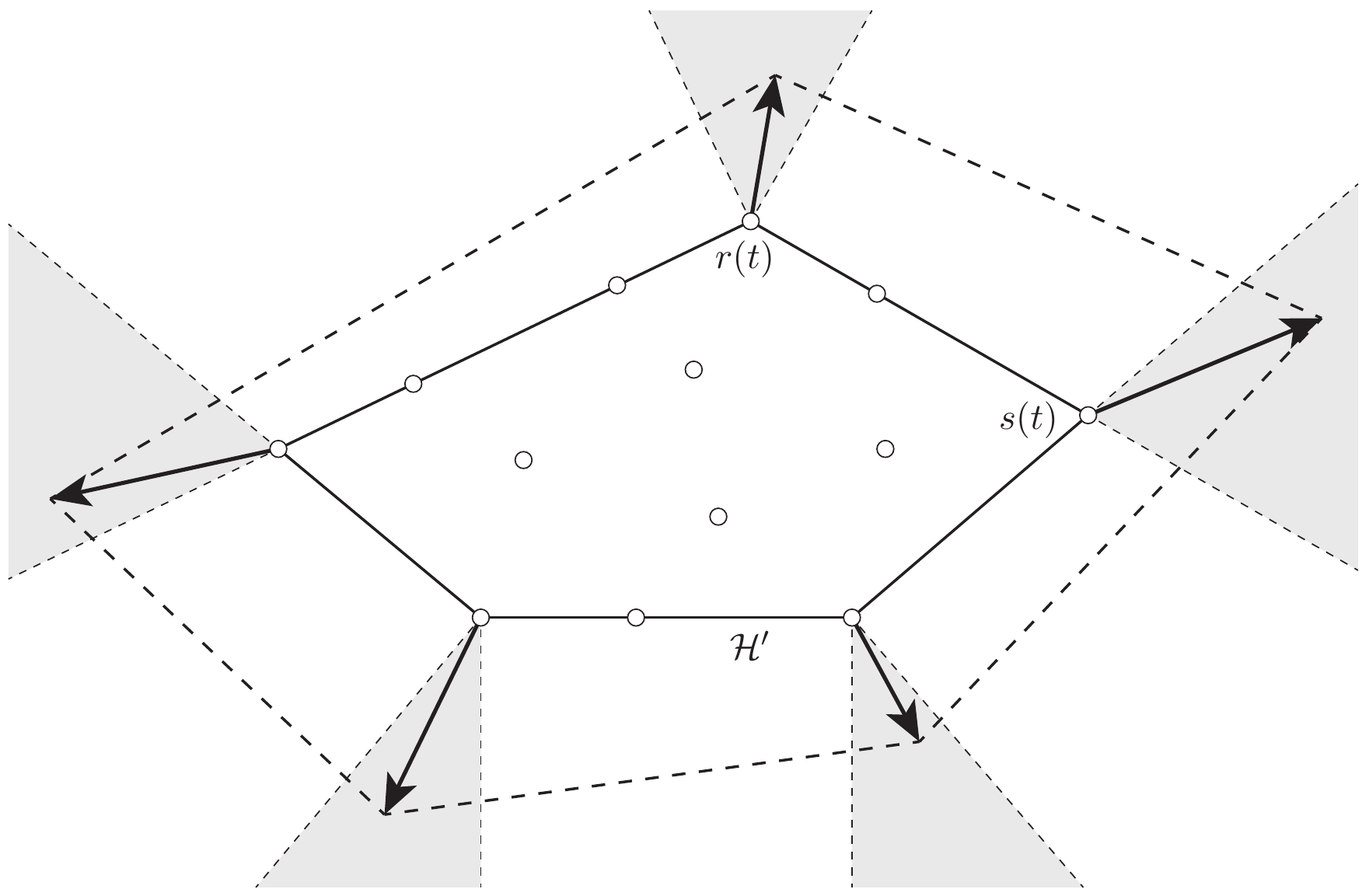}
\caption{Combined motion of internal robots (external robots are not shown)}
\label{fig4}
\end{figure}

\begin{corollary}\label{c:areadiameter}
For $t\geqslant T$, the area of $\mathcal H'(t)$ and the diameter of $\mathcal H'(t)$ do not decrease as $t$ increases.
\end{corollary}
\begin{proof}
By Lemma~\ref{l:intgrow}, if $T\leqslant t_1\leqslant t_2$, then $\mathcal H'(t_1)\subseteq\mathcal H'(t_2)$. Hence the area of $\mathcal H'(t_1)$ cannot be greater than the area of $\mathcal H'(t_2)$, and the diameter of $\mathcal H'(t_1)$ cannot be greater than the diameter of $\mathcal H'(t_2)$.
\end{proof}

\begin{corollary}\label{c:intint}
If $r$ is an internal robot at time $t\geqslant T$, and $r(t)$ is not a non-degenerate vertex of $\mathcal H'(t)$, then $r$ is internal at any time $t'\geqslant t$, and $r(t')$ is not a non-degenerate vertex of $\mathcal H'(t')$.
\end{corollary}
\begin{proof}
By Lemma~\ref{l:intgrow}, if $t'\geqslant t$, then $\mathcal H'(t)\subseteq\mathcal H'(t')$. Moreover, according to line~46 of Algorithm~\ref{alg2}, an internal robot that does not lie at a degenerate vertex of $\mathcal H'$ does not move. It follows that, after time $t$, robot $r$ will not move, hence it will remain internal, and it will never lie at a non-degenerate vertex of $\mathcal H'$.
\end{proof}

Recall that, due to line~53 of Algorithm~\ref{alg2}, a robot on the perimeter of $\mathcal H'$ is able to move only if it completely sees an entire edge of $\mathcal H$. Next we prove that, after time $T$, there exists at least one robot that is able to move.

\begin{lemma}\label{l:canmove}
At any time after $T$, there is a robot that, if activated, makes a non-null movement.
\end{lemma}
\begin{proof}
Suppose for a contradiction that no robot is able to move, and let $p_0$ be a non-degenerate vertex of $\mathcal H'$. If $\mathcal A$ is the internal angle of $\mathcal H'$ at $p_0$, we let $p'_0$ be the point at which the bisector of $\mathbb R^2\setminus \mathcal A$ intersects the perimeter of $\mathcal H$, as Figure~\ref{fig16} shows.

\begin{figure}[h]
\centering
\includegraphics[width=\linewidth]{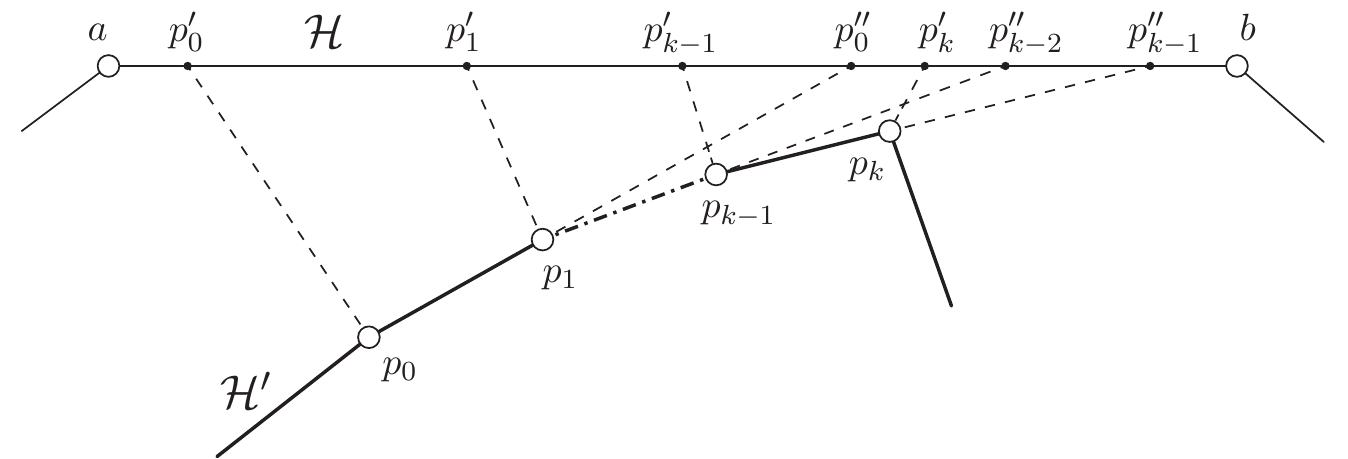}
\caption{The robots in $p_0$, $p_1$, $\cdots$, $p_{k-1}$ are unable to move; the one in $p_k$ is able to move}
\label{fig16}
\end{figure}

Let $ab$ be an edge of $\mathcal H$ such that $p'_0\in ab$ (note that $p'_0$ may coincide with either $a$ or $b$). Since $\mathcal H$ and $\mathcal H'$ are convex, either the segment $ap'_0$ is completely visible to $p_0$, or the segment $bp'_0$ is. Without loss of generality, let $ap'_0$ be completely visible to $p_0$, i.e., $ap'_0\subset \mathbb R^2\setminus\mathcal A$. By definition of $T$, both robots in $a$ and $b$ have their lights let to \emph{External}. Also, note that $p'_0$ lies within $\mathcal A'$, as computed by the robot in $p_0$ executing line~52 of Algorithm~\ref{alg2}. Hence, by line~53, the robot in $p_0$ can move, provided that $b\in \mathbb R^2\setminus \mathcal A$. But by assumption no robot can move, and therefore there must be one robot occupying a non-degenerate vertex of $\mathcal H'$ neighboring $p_0$, say $p_1$, such that the ray from $p_0$ through $p_1$ intersects the segment $ab$, say in $p''_0$.

Observe that $p_1$ is strictly closer to the line $ab$ than $p_0$. By the convexity of $\mathcal H'$, the bisector of the explementary of the internal angle at $p_1$ intersects the segment $ab$, say in $p'_1$. In fact, $p'_1$ lies between $p'_0$ and $p''_0$ (refer to Figure~\ref{fig16}). Also, $p_1$ can see all the points in the segment $ap'_1$. As we argued for $p_0$ in the previous paragraph, there must be another vertex $p_2$ of $\mathcal H'$, closer to the line $ab$, that prevents $p_1$ from seeing the entire segment $ab$.

Proceeding in this fashion, we obtain a sequence $p_0$, $p_1$, $p_2$, $p_3$, $\cdots$ of vertices of $\mathcal H'$, which get closer and closer to the line $ab$. Since these vertices must be all distinct, and there are only finitely many robots in the swarm, there must be one last element of the sequence, $p_k$. It follows that $p_k$ can see all of $ab$, and the corresponding angle bisector intersects $ab$ as well, say in $p'_k$. This implies that the robot in $p_k$ can actually move to a neighborhood of $p'_k$, contradicting our assumption.
\end{proof}

As a consequence of Corollary~\ref{c:intint}, no new robot becomes a non-degenerate vertex of $\mathcal H'$ after time $T$, but some robots may indeed cease to be non-degenerate vertices of $\mathcal H'$, and stop moving forever. Hence, at some time $T'\geqslant T$, the set of robots that lie at non-degenerate vertices of $\mathcal H'$ reaches a minimum $\mathcal M\subset\mathcal R$. By Lemma~\ref{l:intgrow}, the area of $\mathcal H'$ is positive at every time $t\geqslant T'$, and hence $|\mathcal M|\geqslant 3$.

In the next lemma we prove two fundamental properties of $\mathcal H'$ that hold after time $T'$.

\begin{lemma}\label{l:intmisc}
After time $T'$, the following statements hold.
\begin{itemize}
\item[\textbf{\emph{a.}}] The cyclic order of the robots of $\mathcal M$ around $\mathcal H'$ is preserved.
\item[\textbf{\emph{b.}}] The length of the shortest edge of $\mathcal H'(t)$ does not decrease as $t$ increases.
\end{itemize}
\end{lemma}
\begin{proof}
Similarly to Section~\ref{s:ssynch}, we call the gray regions in Figures~\ref{fig3} and~\ref{fig4} the \emph{movement regions} of the respective robots.

Let $t\geqslant T'$, and consider the polygon $\mathcal P$ as defined in the proof of Lemma~\ref{l:intgrow}. By our assumptions on $\mathcal M$, $\mathcal P$ is a convex polygon, or else some robot would cease to occupy a vertex of $\mathcal H'$. Hence $\mathcal P=\mathcal H'(t+1)$ and, by Lemma~\ref{l:intgrow}, $\mathcal H'(t)\subseteq \mathcal H'(t+1)$. Let $r,s\in\mathcal M$ occupy two adjacent edges of $\mathcal H'$ at time $t$, and let $\ell$ be the axis of the segment $r(t)s(t)$. It follows from Algorithm~\ref{alg2} that $\ell$ separates the movement regions of $r$ and $s$ at time $t$ (cf.~Figure~\ref{fig4}). This, paired with the fact that the movement region of a robot of $\mathcal M$ at time $t$ does not intersect the interior of $\mathcal H'$, implies~(a).

Now, to prove~(b), it is sufficient to note that, with the previous paragraph's notation, the distance between the movement regions of $r$ and $s$ at time $t$ is precisely the distance between $r(t)$ and $s(t)$ (see Figure~\ref{fig4}). Therefore, the segment $r(t+1)s(t+1)$ is not shorter than $r(t)s(t)$, implying that the length of the shortest edge of $\mathcal H'$ cannot decrease.
%
\end{proof}

We need one last geometric observation.

\begin{figure}[h]
\centering
\includegraphics[scale=1.2]{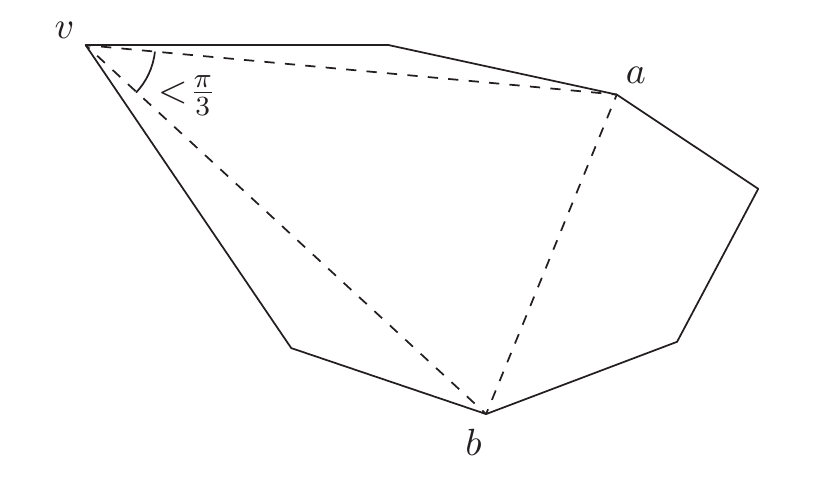}
\caption{The longest edge of triangle $\triangle vab$ is not $ab$}
\label{fig14}
\end{figure}

\begin{lemma}\label{l:diameter}
If the internal angle at vertex $v$ of a convex polygon has measure less than $\pi/3$, then the diameter of the polygon is the distance between $v$ and another vertex.
\end{lemma}
\begin{proof}
The diameter of a polygon is the longest distance between two of its vertices. Suppose for a contradiction that vertices $a$ and $b$ have the maximum distance, with $a\neq v\neq b$, as in Figure~\ref{fig14}. Then, since the polygon is convex, the angle $\angle avb$ is containted in the internal angle at $v$, and therefore its measure is less than $\pi/3$. Since the sum of the internal angles of a triangle is $\pi$, it follows that either $\angle bav>\pi/3$ or $\angle vba>\pi/3$. By the law of sines, in the first case $vb$ is longer than $ab$, and in the second case $va$ is longer than $ab$. In both cases, $ab$ is not the longest segment joining two vertices of the polygon, which is a contradiction.
\end{proof}

We are finally ready to prove the termination, and therefore the correctness, of the interior depletion phase.

\begin{lemma}\label{l:intfinal}
If Algorithm~\ref{alg2} is executed from a non-collinear configuration, after finitely many turns all robots become external, no robot's light is set to \emph{Adjusting}, and no two robots collide.
\end{lemma}
\begin{proof}
The non-collision part has already been proven in Lemma~\ref{l:coll2}, so we need to prove that eventually all robots become external.

By assumption $\mathcal H$ has positive area, and we have to show that all robots become external in finitely many turns. If there is just one internal robot, it keeps moving somewhere within $\mathcal H$, until it either becomes external, or all external robots have been activated. When all external robots have their lights set to \emph{External}, if there is still a single internal robot, it keeps moving toward the same edge of $\mathcal H$, until it finally reaches it.

If there there are exactly two internal robots, they move as shown in Figure~\ref{fig3:a}. It is easy to see that, each time at least one of the two internal robots moves (by at least $\delta>0$), their distance increases by more than $\delta/\sqrt 2$. Therefore, after finitely many turns, at least one of the two internal robots becomes external, and at most one internal robot remains.

Suppose now that there are at least three internal robots, but $\mathcal H'$ has null area, that is, all the internal robots are collinear. Then, according to Algorithm~\ref{alg2}, only the two endpoint robots of $\mathcal H'$ are allowed to move, as Figure~\ref{fig3:a} shows. If they move in such a way that the internal robots keep remaining collinear, eventually one of them reaches the boundary of $\mathcal H$, and there is one less internal robot. Otherwise, they reach a situation in which $\mathcal H'$ has strictly positive area.

Therefore we may assume that $\mathcal H'$ has positive area, and we suppose for a contradiction that some internal robots never become external. By the previous lemmas and observations, we know that at some time $T'$ the situation becomes ``stable''. Specifically, $\mathcal H$ never changes, the set $\mathcal M$ of robots that occupy non-degenerate vertices of $\mathcal H'$ keeps remaining the same, and these robots' positions preserve their order around $\mathcal H'$, by Lemma~\ref{l:intmisc}.a. Also, the area and diameter of $\mathcal H'$ do not decrease, by Corollary~\ref{c:areadiameter}. Let $a(t)$ be the length of the shortest edge of $\mathcal H'(t)$. By Lemma~\ref{l:intmisc}.b, we know that $a(t)$ is a weakly increasing function of $t\geqslant T'$.

Suppose that, at some time $T''\geqslant T'$, some robot $r\in \mathcal M$ that is able to move is activated. By lines~53, 54 of Algorithm~\ref{alg2}, the destination point of $r$ lies on the boundary of $\mathcal H$. Hence, the adversary must stop $r$ before it reaches its destination, or it would become external. But this cannot happen before $r$ has moved by at least $\delta$, implying that $\norm{r(T''+1)-r(T'')}\geqslant\delta$. We distinguish two cases, based on the measure $\alpha$ of the internal angle of $\mathcal H'(T'')$ corresponding to vertex $r(T'')$. We will prove that, if $\alpha<\pi/3$, then the square of the diameter of $\mathcal H'(T'')$ increases by at least a constant; if $\alpha\geqslant \pi/3$, then the area of $\mathcal H'(T'')$ increases by at least a constant.

\begin{figure}[h]
\centering
\includegraphics[width=.85\linewidth]{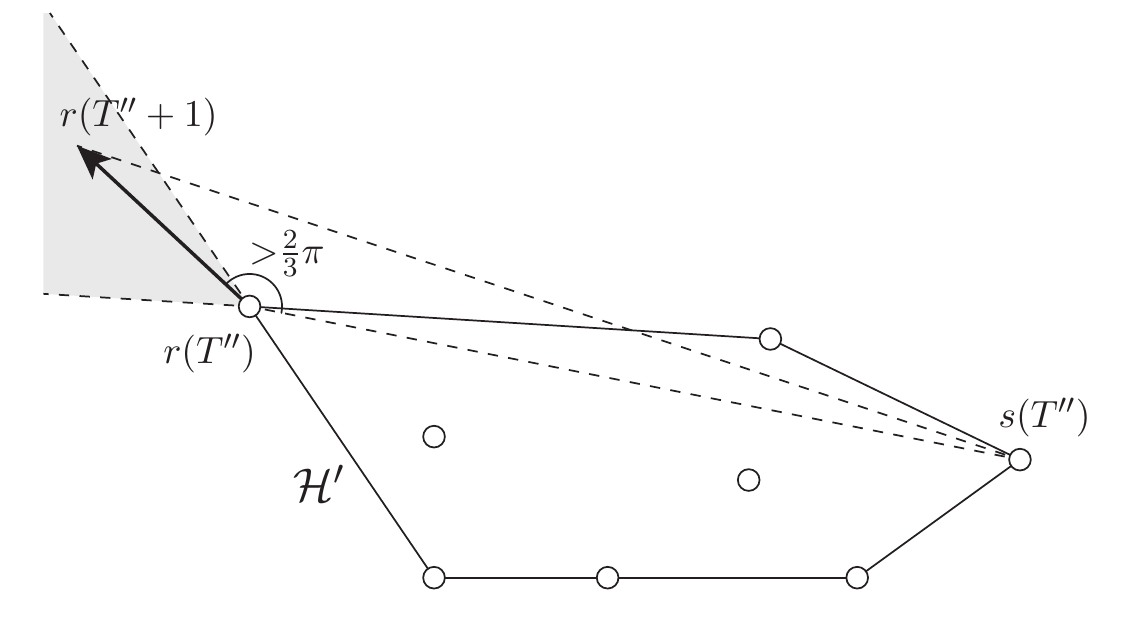}
\caption{The square of the diameter of $\mathcal H'$ increases by at least $\delta^2$}
\label{fig15}
\end{figure}

Suppose that $\alpha<\pi/3$. Let $\mathcal D(t)$ the diameter of $\mathcal H'(t)$, and let $s\in \mathcal M$ be such that $\mathcal D(T'')=\norm{r(T'')-s(T'')}$. Due to line~50 of Algorithm~\ref{alg2}, and referring to Figure~\ref{fig15}, it is easy to prove that
$$\frac 23 \pi<\angle s(T'')r(T'')r(T''+1)\leqslant \pi.$$
Hence
$$\cos(\angle s(T'')r(T'')r(T''+1))<-\frac 12.$$
Because $\mathcal H'(T'')\subseteq \mathcal H'(T''+1)$, it follows that $s(T'')\in \mathcal H'(T''+1)$, and $\mathcal D(T''+1)\geqslant \norm{r(T''+1)-s(T'')}$. Therefore, by the law of cosines applied to triangle $\triangle s(T'')r(T'')r(T''+1)$,
$$\mathcal D^2(T''+1)\geqslant \norm{r(T''+1)-s(T'')}^2> \mathcal D^2(T'')+\delta^2+\mathcal D(T'')\cdot\delta>\mathcal D^2(T'')+\delta^2.$$
Hence, in this case, the square of the diameter of $\mathcal H'$ increases by at least $\delta^2$.

\begin{figure}[h]
\centering
\includegraphics[width=.85\linewidth]{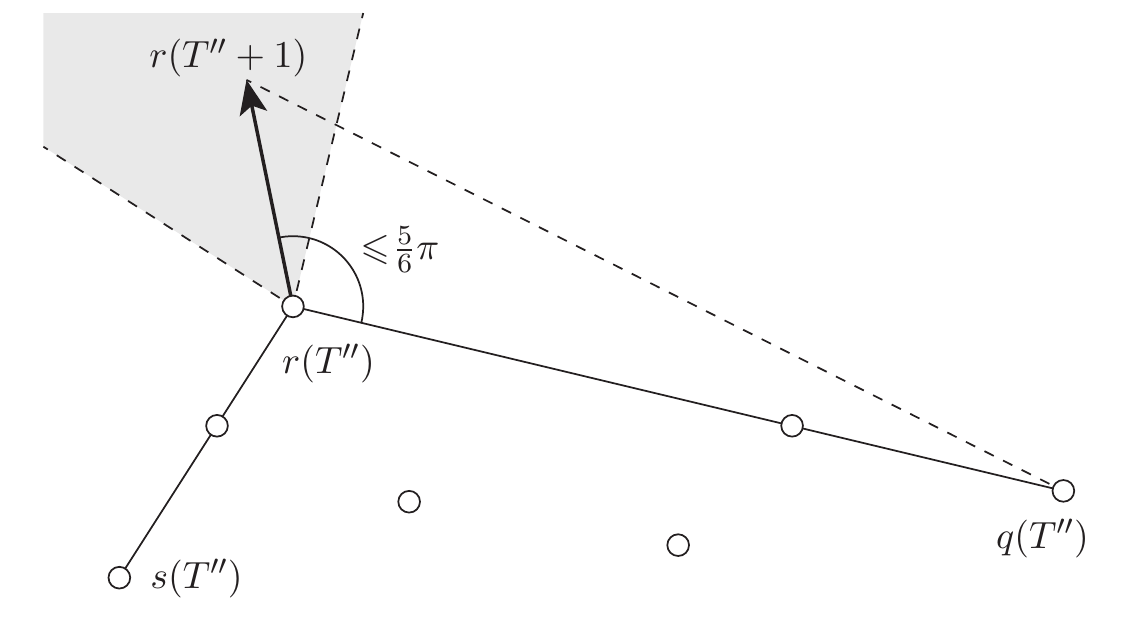}
\caption{The area of $\mathcal H'$ increases by at least $a(T')\cdot\delta/8$}
\label{fig5}
\end{figure}

Let $\alpha\geqslant \pi/3$, and let $q,s\in\mathcal M$ occupy the two vertices of $\mathcal H'(T'')$ adjacent to vertex $r(T'')$, in such a way that
$$\angle q(T'')r(T'')r(T''+1)\leqslant \angle r(T''+1)r(T'')s(T''),$$
as Figure~\ref{fig5} shows. It follows that
$$\frac \pi 2\leqslant \angle q(T'')r(T'')r(T''+1)\leqslant \frac{2\pi-\alpha}{2}\leqslant \frac 56 \pi.$$
Then, recalling that $\mathcal H'(T'')\subseteq \mathcal H'(T''+1)$, we have that the area of $\mathcal H'$ increases at least by the area of the triangle $\triangle q(T'')r(T'')r(T''+1)$, which in turn is at least
$$\frac 12\cdot\norm{q(T'')-r(T'')}\cdot\norm{r(T''+1)-r(T'')}\cdot\sin(\angle q(T'')r(T'')r(T''+1))$$
$$\geqslant \frac 12\cdot a(T'')\cdot\delta\cdot\sin(5\pi/6)\geqslant a(T')\cdot\delta/8.$$

Concluding, every time a robot of $\mathcal M$ moves, either the square of the diameter of $\mathcal H'$ increases by at least $\delta^2$, or the area of $\mathcal H'$ increases by at least $a(T')\cdot\delta/8$. But Lemma~\ref{l:canmove} states that there is always at least one robot of $\mathcal M$ that is able to move, and therefore the robots will move infinitely many times, due to the fairness of the \SSynch scheduler. From Corollary~\ref{c:areadiameter}, and from the fact $\delta^2$ and $a(T')\cdot\delta/8$ are positive constants, it follows that either the diameter or the area of $\mathcal H'$ grows unboundedly. This contradicts the fact that $\mathcal H'(t)\subseteq \mathcal H(t)$, and that $\mathcal H(t)$ is independent of $t$.
\end{proof}

\subsubsection{Vertex Adjustments Phase}

Due to Lemma~\ref{l:intfinal}, all robots become external at some point, and it remains to show that they finally reach a \emph{strictly} convex configuration and correctly terminate. This turns out to be a significantly easier task.

\begin{lemma}\label{l:extfinal}
If Algorithm~\ref{alg2} is executed from a configuration in which all robots are external and no robot's light is set to \emph{Adjusting}, then after finitely many turns all robots have terminated, no two of them have collided, and the configuration is strictly convex.
\end{lemma}
\begin{proof}
In the vertex adjustments phase, each robot eventually sets its own light to \emph{External} (if it has not already done so in the interior depletion phase). Meanwhile, as soon as a vertex robot $r$ sees only robots whose lights are set to \emph{External}, it sets its own light to \emph{Adjusting} and makes a default move, as in Figure~\ref{fig1:b}. Recall that robots are \NRigid, hence a vertex robot may be stopped before reaching its destination, but not before having moved by at least $\delta>0$. When $r$ is activated again, it sees itself in the \emph{Adjusting} state and, if it sees a robot set to \emph{Off}, it necessarily also sees a robot set to \emph{External}. Indeed, after the default move, $r$ can see all the robots in the swarm. Note that, if a robot is set to \emph{Adjusting}, neither of its two neighbors on the perimeter of the convex hull can be \emph{Off}. Hence, if $r$ sees a robot set to \emph{Off}, and since it sees itself set to \emph{Adjusting}, it must also see a robot set to \emph{External}. Therefore $r$ sets its light back to \emph{External} (lines~21--23), thus allowing other robots to move. $r$ also terminates because, as already noted, its two neighbors cannot be \emph{Off} (line~24) and, since $r$ sees every robot in the swarm and the configuration is convex, $r$ does not see any internal robots (line~25).

As observed in Section~\ref{s:ssynch}, where a similar procedure was used to reduce the size of the convex hull while increasing the set of vertex robots, when a robot occupying a vertex of the convex hull moves, it becomes a vertex of the new convex hull. On the other hand, if a robot $r$ lies in the interior of an edge of the convex hull and one of its two neighbors $s$ lies on a vertex, then, as soon as $s$ moves, $r$ becomes a vertex of the new convex hull. Hence, eventually all robots become vertices of the convex hull. After that, whenever a robot is activated, it permanently sets its light to \emph{External} and terminates. It follows that eventually all robots terminate in a strictly convex configuration.

Moreover, by the observations already made in Section~\ref{s:invariants}, and referring to Figure~\ref{fig2}, it is clear that no collisions can occur in this phase.
\end{proof}

\subsubsection{Segment Breaking Phase}\label{s:segbreak}

From Lemmas~\ref{l:intfinal} and~\ref{l:extfinal}, the correctness of Algorithm~\ref{alg2} immediately follows, provided that the robots are not initially collinear. This case is considered in the following.

\begin{theorem}\label{t:correctness2}
Algorithm~\ref{alg2} solves \Obst for \NRigid\ \SSynch robots with 3-colored lights.
\end{theorem}
\begin{proof}
Due to Lemmas~\ref{l:intfinal} and~\ref{l:extfinal}, we only have to show that the case in which the robots are initially collinear correctly evolves into one of the other cases. If $n\leqslant 3$, this is easy to verify through a case analysis, following the algorithm's description of Section~\ref{s:alg2descr}.

So, let $n\geqslant 4$, and let robots $r$ and $s$ initially occupy the vertices of the line segment $\mathcal H(0)$. Nothing happens until $r$ or $s$ is activated; then, at least one of them becomes \emph{Adjusting} and moves orthogonally to $\mathcal H$. After $r$ or $s$ has moved, some robots eventually become \emph{External} (line~35). Indeed, if both $r$ and $s$ move, then all other robots can see both of them. In particular, there is at least one such robot occupying a non-degenerate vertex of the convex hull, which sees two \emph{Adjusting} robots, and therefore becomes \emph{External} (line~33). If only $r$ moves (or vice versa), then $s$ sees only three robots (including itself), and its corresponding convex hull angle is acute, hence it becomes \emph{External} (line~32). In the latter case, only $s$ is allowed to become \emph{External}, while the other robots wait, because their corresponding angles are not acute. When a robot has become \emph{External}, any robot that was \emph{Adjusting} can see it, and hence becomes \emph{External} as well. If, in addition, neither of its two neighbors is \emph{Off}, it also terminates (note that an \emph{Adjusting} robot must be located on the convex hull's boundary).

At this point the execution proceeds normally, except that there may be one of two vertex robots that have already terminated, and we have to show that this does not prevent the others from forming a strictly convex configuration. Obviously the interior depletion phase causes no trouble and it is carried out correctly, but the vertex adjustments phase might ``get stuck''. We will prove that this is not the case. Clearly, if at most one robot has terminated, all robots except perhaps one are able to move in the vertex adjustments phase, and therefore they all become non-degenerate vertices. If $r$ and $s$ initially move in the same direction, they become neighboring vertices, and all the other robots become consecutive external robots. It is easy to see that in this case the external robots are still able to make adjusting movements in cascade, and become non-degenerate vertices.

\begin{figure}[h]
\centering
\subfigure[$s'$ becomes internal]{\label{fig10b:a}\includegraphics[width=.715\linewidth]{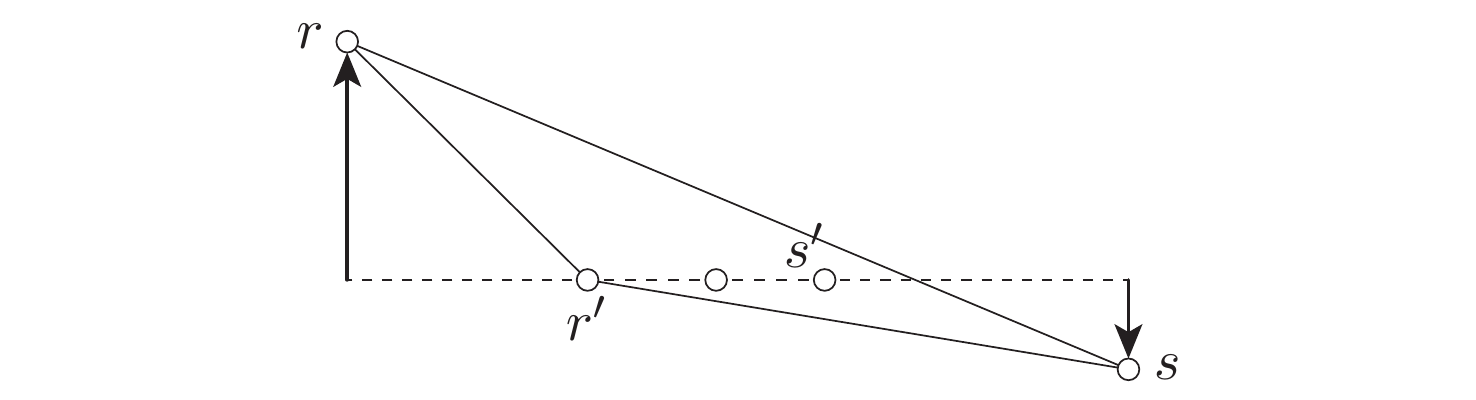}}
\subfigure[$s'$ becomes a degenerate vertex of the convex hull]{\label{fig10b:b}\includegraphics[width=.715\linewidth]{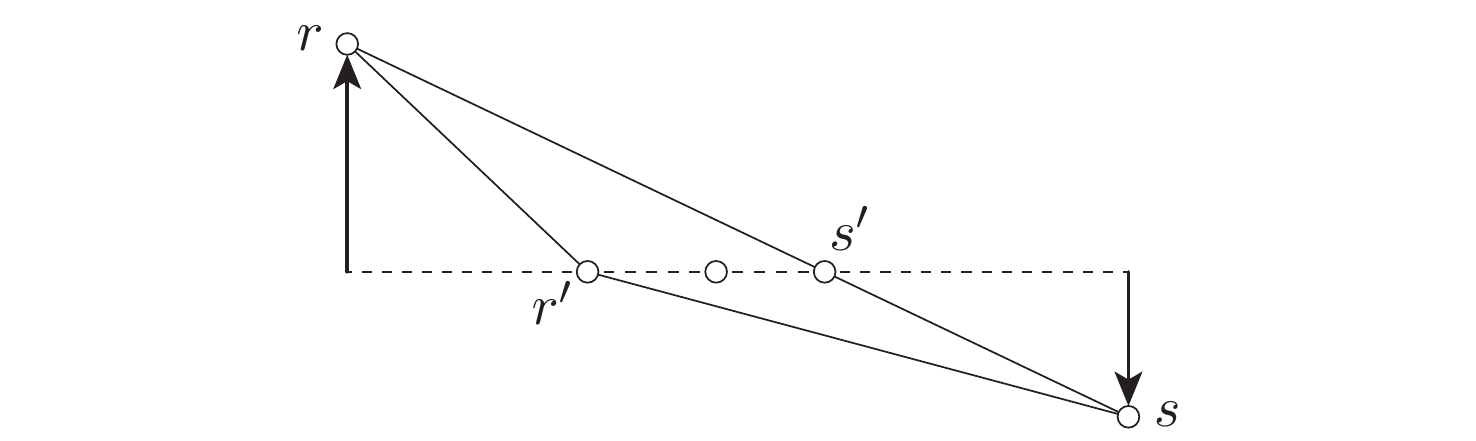}}
\subfigure[$s'$ becomes a non-degenerate vertex of the convex hull]{\label{fig10b:c}\includegraphics[width=.715\linewidth]{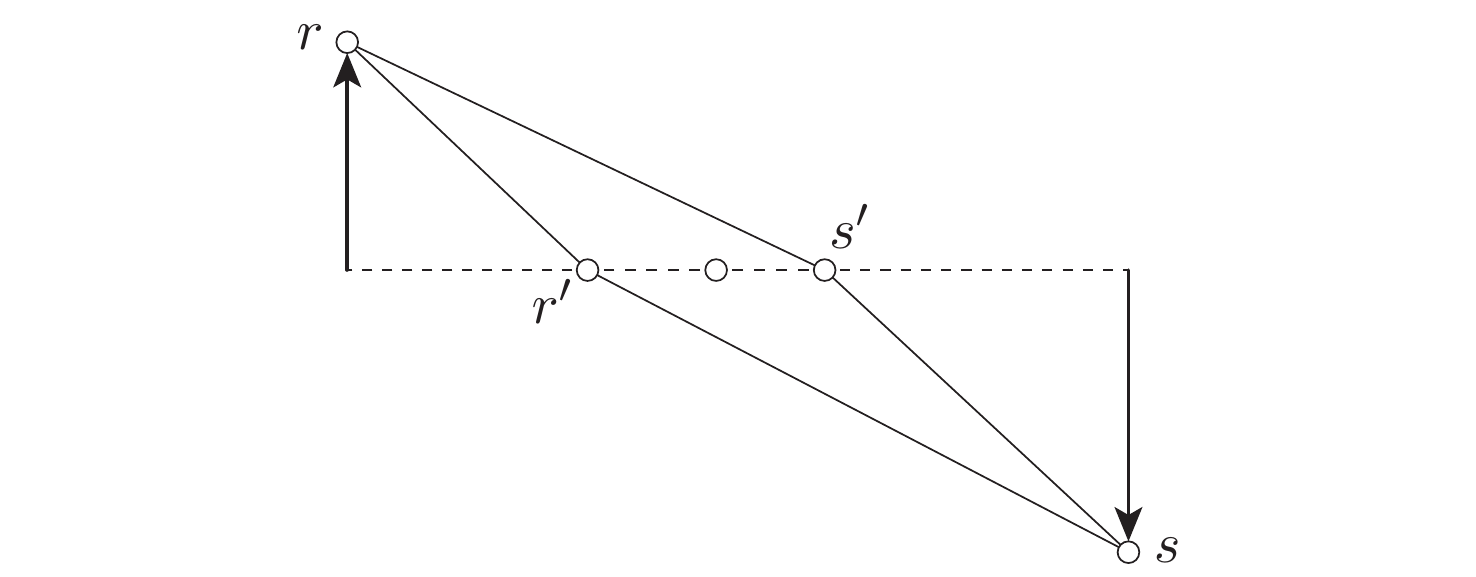}}
\caption{Possible evolutions of a collinear configuration in which the endpoints move in opposite directions}
\label{fig10b}
\end{figure}

Finally, suppose that $r$ and $s$ initially move in opposite directions. Let $r'$ be the robot closest to $r$, and $s'$ the robot closest to $s$. Since $n\geqslant 4$, $r'\neq s'$. As already noted, at least one between $r'$ and $s'$ becomes a non-degenerate vertex of the convex hull, say $r'$. On the other hand, depending on how much $r$ and $s$ move, $s'$ may become internal, or a degenerate or non-degenerate vertex of the convex hull. If $s'$ becomes internal, as in Figure~\ref{fig10b:a}, then both $r$ and $s$ see an internal robot: indeed, after the move, $r$ and $s$ can see all robots. Then, $r$ and $s$ become \emph{External} but do not terminate (line~25), and everything works as intended in the later phases. If $s'$ becomes a degenerate vertex of the convex hull, as in Figure~\ref{fig10b:b}, then it is an \emph{Off} neighbor of both $r$ and $s$, which once again become \emph{External} but do not terminate (line~24). Again, the execution transitions seamlessly into another phase. If $s'$ becomes a non-degenerate vertex of the convex hull, as well as $r'$, then all the robots between them become internal. Note that, in this case, $r$ and $s$ may terminate (indeed, they may not be able to see each other, and hence they may not realize that there are internal robots), but they do not lie at adjacent vertices of $\mathcal H$, due to the presence $r'$ and $s'$. After the interior depletion phase, $r'$ and $s'$ are able to adjust, thus enabling all other external robots to become non-degenerate vertices, in cascade.
\end{proof}

\section{Solving \Obst for \ASynch Robots}\label{s:asynch}

In this section we briefly touch on the \ASynch model. In the \Rigid case, we show that Algorithm~\ref{alg2} solves the \Obst problem. In the \NRigid case, we show how to solve \Obst assuming that the robots agree on the direction of one coordinate axis.

\subsection{\Rigid\ \ASynch Robots}\label{s:rasynch}

Algorithm~\ref{alg2} turns out to solve the \Obst problem for \Rigid\ \ASynch robots, as well. For the interior depletion phase, the collision avoidance proof gets slightly more complex, but termination is easier to prove. On the other hand, the vertex adjustments phase and the segment breaking phase work almost in the same way.

First we state an equivalent of Lemma~\ref{l:coll1}. The only difference is that, instead of a generic time $t\in\mathbb N$, now we have to consider a specific time $t\in\mathbb R$ at which a robot $r$ performs a \Look. Also, instead of considering the position of $r$ at time $t+1$, we consider the destination point computed after such a \Look. After these changes, the proof of Lemma~\ref{l:coll1} works in the \ASynch case as well, and therefore we have the following.

\begin{lemma}\label{l:coll3}
Let \Rigid\ \ASynch robots execute Algorithm~\ref{alg2}, and let $r$ and $s$ be two internal robots at time $t\in\mathbb R$. If $r$ executes a \Look phase at time $t$, and the next destination point of $r$ is $p$, then
$$(p-r(t))\bullet(s(t)-r(t))\leqslant 0.$$
\hfill\qed
\end{lemma}

The previous lemma can be used to prove that no collisions occur during the interior depletion phase.

\begin{lemma}\label{l:coll3b}
If \Rigid\ \ASynch robots execute Algorithm~\ref{alg2} from a non-collinear configuration, no collisions occur as long as there are internal robots.
\end{lemma}
\begin{proof}
Suppose for a contradiction that the internal robot $r$ performs a \Look at time $t$, then robot $s$  performs a \Look at time $t'\geqslant t$, and they collide at time $t''\geqslant t'$, in $r(t'')=s(t'')$. We may further assume that this is the first collision between the two robots, and therefore $r(t)\neq s(t)$. Because the model is \Rigid and each internal robot's destination point is on the convex hull, it follows that each internal robot makes exactly one move and then becomes external. Therefore, $r(t)$, $r(t')$, $r(t'')$, and the destination point of $r$ are all collinear, and the same holds for $s$. Additionally, we have $s(t)=s(t')$ (see Figure~\ref{fig9}).

\begin{figure}[h]
\centering
\includegraphics[width=.65\linewidth]{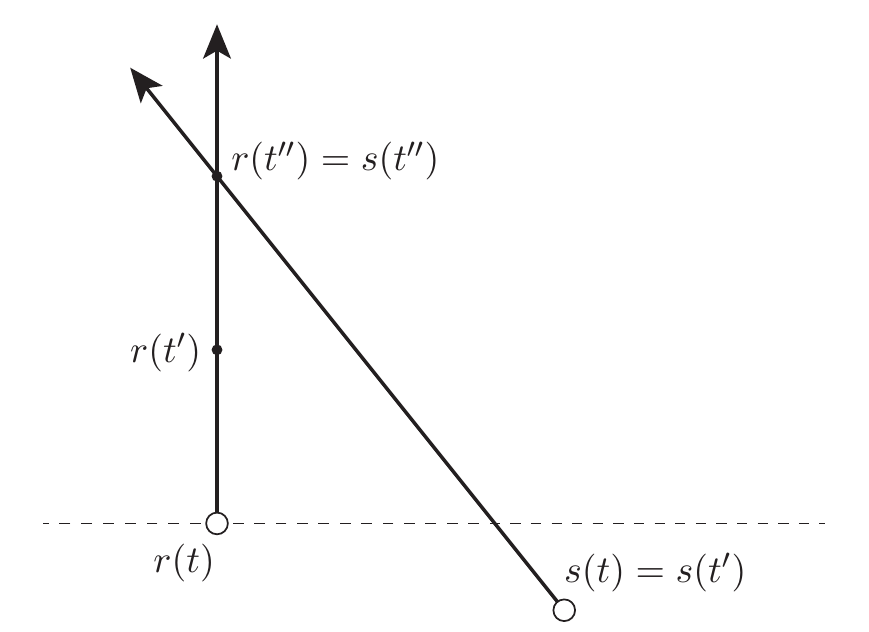}
\caption{Two colliding internal robots}
\label{fig9}
\end{figure}

By Lemma~\ref{l:coll3} applied to $s$ at time $t'$, and because $s(t'')$ lies between $s(t')$ and the destination point of $s$, we have
$$(s(t'')-s(t'))\bullet (r(t')-s(t'))\leqslant 0.$$
On the other hand, $\norm{s(t'')-s(t')}\geqslant 0$, implying that
$$(s(t'')-s(t'))\bullet (s(t')-s(t''))\leqslant 0.$$
By adding the two inequalities together, we obtain
$$(s(t'')-s(t'))\bullet (r(t')-s(t''))\leqslant 0.$$
Recall that $s(t'')=r(t'')$ and that $r(t')$ lies between $r(t)$ and $r(t'')$, and therefore the last inequality implies
$$(s(t'')-s(t'))\bullet (r(t)-r(t'))\leqslant 0,$$
hence
$$(s(t'')-s(t'))\bullet (r(t)-s(t')+s(t')-r(t'))\leqslant 0,$$
and
$$(s(t'')-s(t'))\bullet (r(t)-s(t'))\leqslant (s(t'')-s(t'))\bullet (r(t')-s(t')).$$
But we already know that the right-hand side is not positive, hence so is the left-hand side:
$$(s(t'')-s(t'))\bullet (r(t)-s(t'))\leqslant 0.$$
Now, by Lemma~\ref{l:coll3} applied to $r$ at time $t$, and recalling that $r(t'')$ lies between $r(t)$ and the destination point of $r$, we have
$$(r(t'')-r(t))\bullet (s(t)-r(t))\leqslant 0.$$
If we add together the last two inequalities and we recall that $s(t')=s(t)$, we get
$$(r(t)-s(t))\bullet (r(t)-r(t'')+s(t'')-s(t'))\leqslant 0.$$
Because $r(t'')=s(t'')$ and $s(t')=s(t)$, we finally obtain
$$(r(t)-s(t))\bullet (r(t)-s(t))\leqslant 0,$$
which is equivalent to $\norm{r(t)-s(t)}\leqslant 0$, implying that $r(t)=s(t)$, a contradiction.
\end{proof}

We can now prove that Algorithm~\ref{alg2} works also with \Rigid\ \ASynch robots.

\begin{theorem}\label{t:correctness3}
Algorithm~\ref{alg2} solves \Obst for \Rigid\ \ASynch robots with 3-colored lights.
\end{theorem}
\begin{proof}
In the interior depletion phase there can be no collisions, due to Lemma~\ref{l:coll3b}. Also, whenever an internal robot moves, its destination lies on the boundary of the ``real'' $\mathcal H$ (cf.~line~53 of Algorithm~\ref{alg2}). Since movements are rigid, such a robot becomes external in a single move. Suppose for a contradiction that the interior depletion phase does not terminate. Then, at some point, the set of external robots reaches a maximum, all the external robots are set to \emph{External}, and no robot is moving (indeed, moving internal robots are bound to become external). Hence, Lemma~\ref{l:canmove} can be applied. As a consequence, there is some internal robot that is able to move, and which will therefore reach the convex hull's perimeter at the end of its next \Move phase, thus becoming external. This contradicts our assumptions. Therefore, in finite time all robots become external, and the interior depletion phase terminates.

When all robots are external, none of them moves unless it sees only robots set to \emph{External} (line~27). This means that, in the vertex adjustments phase, a robot waits until it is sure that no robot is in the middle of a move (note that this holds also for robots that it cannot see, because as soon as one of them moves it becomes visible to all other robots). Indeed, a robot sets itself to \emph{Adjusting} right before starting to move and sets itself back to \emph{External} when it is done moving. Hence the robots synchronize themselves, and we may pretend them to be \SSynch, as opposed to \ASynch. Then, the proof proceeds exactly as in Lemma~\ref{l:extfinal}.

\begin{figure}[h]
\centering
\subfigure[$r$ and $s$ move in the same direction]{\label{fig10:a}\includegraphics[width=.75\linewidth]{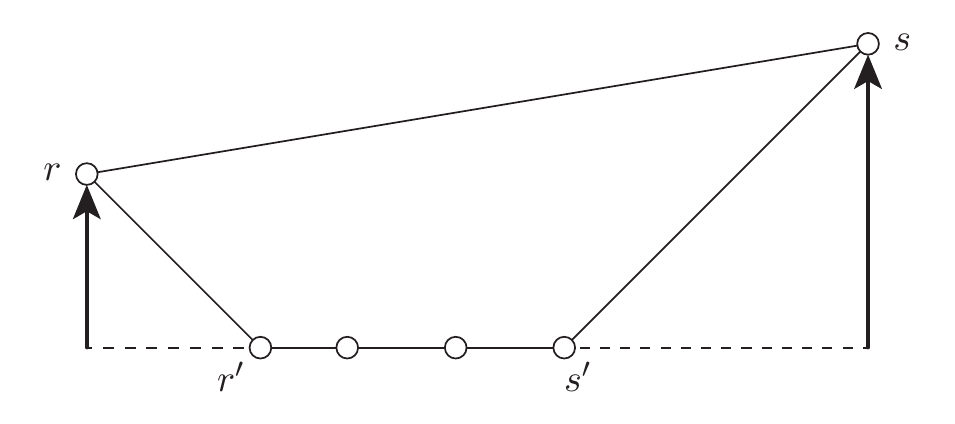}}
\subfigure[$r$ and $s$ move in opposite directions]{\label{fig10:b}\includegraphics[width=.75\linewidth]{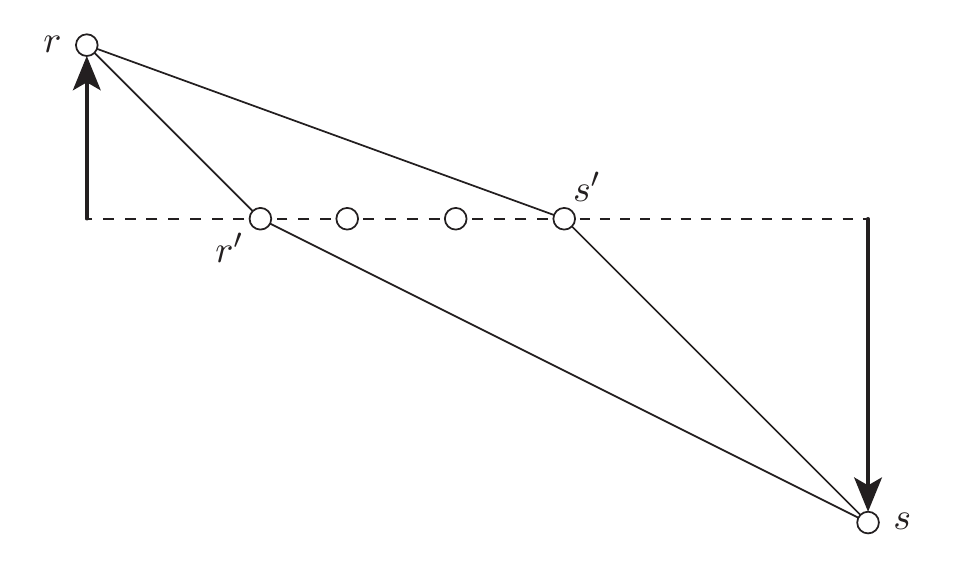}}
\caption{Evolutions of a collinear configuration in which both endpoints move}
\label{fig10}
\end{figure}

In the case in which the robots are initially collinear, the proof follows the lines of Theorem~\ref{t:correctness2}, with a few differences. Indeed, despite being \ASynch, the robots manage to wait for each other and synchronize their actions. Suppose that one endpoint robot $r$ becomes \emph{Adjusting} and starts moving to its destination. Then, every robot is bound to wait for the other endpoint robot, $s$, to take action. So, $s$ could either become \emph{Adjusting} as well and start moving (if it performed its \Look before $r$ started moving), or it could notice $r$ and become \emph{External}. If $r$ and $s$ are both \emph{Adjusting} and moving asynchronously, some other robots eventually become \emph{External}, but do not move yet. In particular, referring to Figure~\ref{fig10}, at least robots $r'$ and $s'$ can become external in this phase. Notice that, if $r$ and $s$ move asynchronously in opposite directions (Figure~\ref{fig10:b}), $r'$ and $s'$ may switch between being internal and being external several times. However, as soon as they set their light to \emph{External}, they do not set it back to \emph{Off} even if they become internal again. But $r$ moves exactly by $r(0)r'(0)$, and $s$ moves exactly by $s(0)s'(0)$ (line~12), because the model is \Rigid. This movement length is chosen in such a way that both $r'$ and $s'$ eventually become vertex robots, as Figure~\ref{fig10:b} suggests. Therefore the colors of $r'$ and $s'$ are eventually consistent, despite asynchrony. So, every robot waits for both $r$ and $s$ to see some \emph{External} robots and thus become \emph{External} themselves. Only then do other robots start moving (lines~26--29). As a consequence, we may once again pretend that the robots in this phase are \SSynch, and the proof is completed as in Theorem~\ref{t:correctness2}.
\end{proof}

\subsection{\NRigid\ \ASynch Robots with Agreement on One Axis}
\label{sec:axis}

Unfortunately, for \NRigid\ \ASynch robots, our correctness proof of the interior depletion phase of Algorithm~\ref{alg2} fails. Indeed, to prove collision avoidance, we used in a crucial way the fact that, if two internal robots are moving at the same time, then at most one of them saw the other robot in the middle of a movement. This is true under the \NRigid\ \SSynch model (obviously) and under the \Rigid\ \ASynch model (because each internal robot becomes external after only one move), but not under the \NRigid\ \ASynch model. In this model, an internal robot $r$ may perform different moves in different directions before becoming external. For instance, if $r$'s first movement is stopped by the adversary and, in the meantime, new robots have become \emph{External} or new robots have become visible, $r$ may decide to move in a significantly different direction the second time. This, paired with the ability of the \ASynch scheduler to hold a moving robot for an indefinitely long time and then release it and let it complete its move, does cause collisions in some (quite pathological) cases. On the other hand, however, the vertex adjustments phase of Algorithm~\ref{alg2} works under all models; therefore we only need to replace the interior depletion phase and the segment breaking phase.

With the additional assumption that all robots agree on one axis, there is an easy way to fix the interior depletion phase, which is illustrated in Figure~\ref{fig6}. Say that the robots agree on the $y$ axis, i.e., they agree on the ``North'' direction, but they may disagree on ``East'' and ``West''. Then, if an internal robot sees that every robot that lies to the North is set to \emph{External} (i.e., if its own $y$ coordinate is maximum among the internal robots), it is eligible to move. If there is a row of several internal robots that are all eligible to move (as in Figure~\ref{fig6}), then only the two endpoints are allowed to move, and the others wait. The left endpoint moves to the upper-left quadrant, and the right endpoint moves to the upper-right quadrant, and their destination points are on the convex hull, but not on locations already occupied by external robots. To guarantee termination, we make each robot move straight to the North toward the boundary of the convex hull of the visible robots, unless there are external robots in the way. In this special case, we make the robot move slightly sideways.

\begin{figure}[h]
\centering
\includegraphics[width=.8\linewidth]{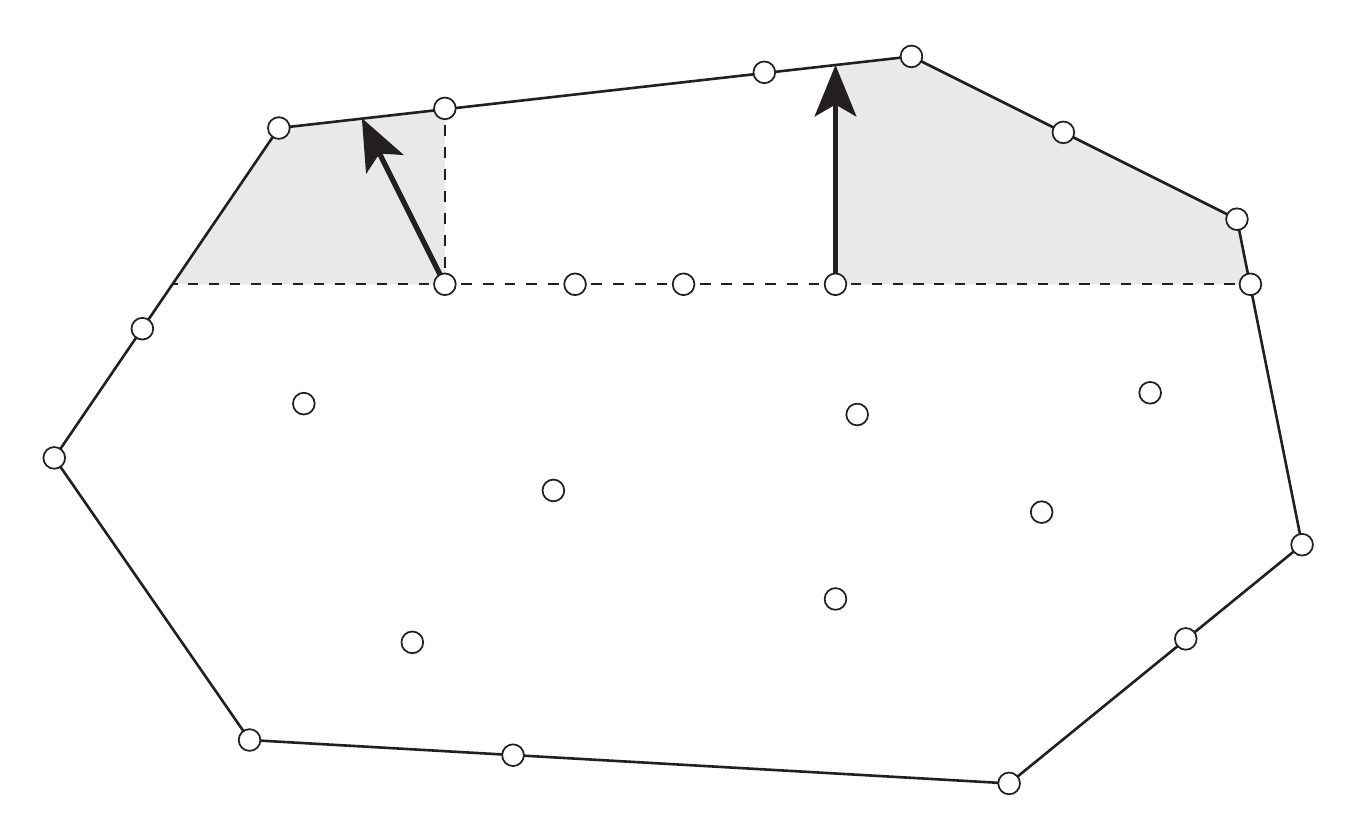}
\caption{Interior depletion with agreement on one axis}
\label{fig6}
\end{figure}

Also the protocol for the segment breaking phase needs some modifications: indeed, referring to Figure~\ref{fig10:b}, in which $r$ and $s$ move in opposite directions, it is no longer true that $r'$ and $s'$ will eventually be external robots when $r$ and $s$ stop (recall that robots are \NRigid now). Unfortunately, $r'$ and $s'$ may become temporarily external while $r$ and $s$ move, and thus they may (permanently) set themselves to \emph{External}, which could lead to inconsistent behaviors. Once again, we can fix the protocol if the robots agree on the $y$ axis: now, in the segment breaking phase, an endpoint robot moves according to Algorithm~\ref{alg2} only if it has the maximum $y$ coordinate. This makes only one endpoint move in most cases, which eliminates the aforementioned issue. Moreover, if both endpoints have the same $y$ coordinate, they will both move North, thus forming a configuration like the one in Figure~\ref{fig10:a}, which causes no trouble.

\begin{theorem}
The \Obst problem is solvable by \NRigid\ \ASynch robots carrying 3-colored lights, provided that they agree on one axis.
\end{theorem}
\begin{proof}
We show that the above algorithm is correct. In the interior depletion phase, there can be no collisions, and each internal robot eventually reaches the convex hull. Indeed, suppose that initially there is a unique internal robot $r$ with largest $y$ coordinate. As soon as enough external robots have set themselves to \emph{External}, $r$ starts moving North, and no other robot moves. Eventually $r$ becomes external without colliding with any robot (note that, even if $r$ does not initially see the boundary of the convex hull, it will eventually see it after finitely many moves).

If several internal robots have the largest $y$ coordinate, as in Figure~\ref{fig6}, the argument is similar. At most two robots can move at the same time, and they cannot collide because the difference of their $x$ coordinates cannot decrease. After enough cycles, either they have reached the convex hull, or one of them has been ``left behind'' and is no longer eligible to move. Either way, at least one internal robot eventually becomes external.

Once an internal robot has become external, the same argument repeats for all other internal robots. Note that these ``sub-phases'' do not interfere with each other, because a new robot becomes eligible to move only after the previous eligible robots have stopped on the convex hull.

The moment the last internal robot becomes external, no robot is moving, and therefore the whole swarm correctly transitions to the vertex adjustments phase, which works exactly as described in Theorem~\ref{t:correctness3} and Lemma~\ref{l:extfinal}.

If the robots are initially collinear, they correctly transition to a non-collinear configuration, as in Theorem~\ref{t:correctness3}. Indeed, note that the two endpoint robots cannot move in opposite directions (as in Figure~\ref{fig10:b}), and hence it does not matter if they are \Rigid or \NRigid, since in this case it is not harmful if they move by smaller amounts than those indicated by Algorithm~\ref{alg2} (cf.~Figure~\ref{fig10:a}). The same clearly holds if $n=3$ and the middle robot executes line~16.
\end{proof}

\section{Related Problems and Alternative Models}\label{s:extra}

Here we discuss some applications of the previous \Obst algorithms to other problems, and we also discuss different robot models.

\subsection{Forming a Convex Configuration}
As already observed, Algorithm~\ref{alg1} also solves the \Conv problem, where the robots have to terminate in a strictly convex position. Moreover, no robot ever crosses the perimeter of the initial convex hull unless, of course, all the robots are initially collinear. This works for \Rigid\ \SSynch robots carrying 2-colored lights.

For \NRigid\ \SSynch robots carrying 3-colored lights, Algorithm~\ref{alg2} also solves the \Conv problem, but it has an additional property: during the interior depletion phase, the convex hull of the robots remains unaltered (unless all robots are collinear), and in the vertex adjustments phase it shrinks a little, due to the small movements of the vertices. We can actually make these movements as small as we want, by changing line~29 of Algorithm~\ref{alg2} into
$$\mbox{Move to }(a.\mbox{\emph{position}}+b.\mbox{\emph{position}})\cdot\frac\varepsilon{\norm{a.\mbox{\emph{position}}+b.\mbox{\emph{position}}}},$$
where $\varepsilon$ is an arbitrarily-chosen positive constant. Similarly, in lines~12 and~16 we can make the robot move orthogonally to $\mathcal H$ by $\varepsilon$ or less. As a result, we can guarantee that the robots will terminate in a (strictly convex) configuration whose vertices are contained in an $\varepsilon$-wide band around the initial convex hull's perimeter.

Similar observations hold for the algorithms and models discussed in Section~\ref{s:asynch}.

\subsection{Forming a Circle}\label{s:circle}

As a followup to Algorithms~\ref{alg1} and~\ref{alg2}, the robots can even solve the \Circ problem, in which they have to become concircular and then terminate. Moreover, if they are \Rigid\ \SSynch (respectively, \NRigid\ \SSynch), they can do so with the same 2-colored (respectively, 3-colored) lights that they used to solve \Obst.

First, it is necessary to slightly modify the termination condition of the algorithms: in Algorithm~\ref{alg1}, when a robot sees only robots set to \emph{Vertex}, it does not terminate, but it starts executing a \emph{circle formation} phase. Similarly, in Algorithm~\ref{alg2}, we remove lines~23--25, thus preventing vertex robots from reverting their color to \emph{External} and terminating after they have adjusted their position. Instead, they wait until they only see robots set to \emph{Adjusting}. Accordingly, in lines~27 and~36 we remove the conditions that prevent robots from moving if they see other robots set to \emph{Adjusting}. Since we are assuming that robots are \SSynch, it is straightforward to see that the correctness proof of Section~\ref{s:ssynch2} goes through even after these modifications to the protocol, and that eventually all robots are set to \emph{Adjusting}. At this point, the circle formation phase starts.

Since all robots are set to \emph{Adjusting}, each robot knows that all of them occupy non-degenerate vertices of the convex hull, and that there are no other robots in the swarm. Hence the phase starts in a strictly convex configuration, and all the robots see each other. In particular, the \emph{Smallest Enclosing Circle} (SEC) computed by each active robot is the same. In the circle formation phase, the robots move toward the perimeter of the SEC in a precise order, as illustrated in Figure~\ref{fig7}. If a robot lies in $p$, which is not on the SEC, and one of its neighbors lies in $s$, which is on the SEC, then the robot in $p$ moves along the extension of the edge of the convex hull that is incident to $p$ and not to $s$. If both neighbors of the robot lie on the SEC (as with the robot in $q$ in Figure~\ref{fig7}), it chooses one of its two incident edges, and moves along the extension of that edge.

\begin{figure}[h]
\centering
\includegraphics[width=.6\linewidth]{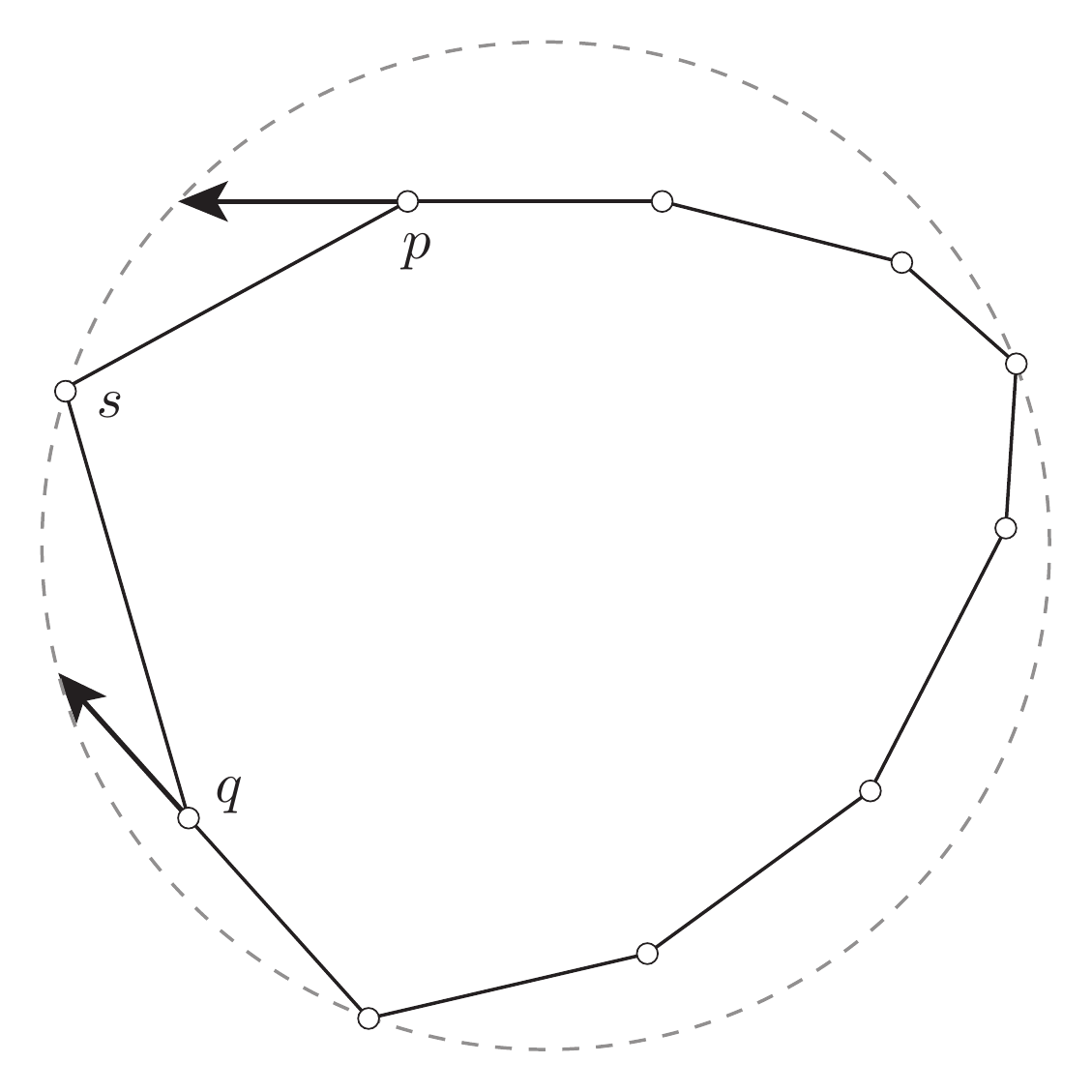}
\caption{Forming a circle}
\label{fig7}
\end{figure}

It is clear that the combined motion of the robots does not cause collisions or obstructions, and that the SEC is always preserved. Indeed, any robot that is already on the SEC remains still, and those that are inside the SEC are allowed to move only within the SEC itself. Moreover, the direction in which each robot moves is preserved until one of them reaches the SEC. Hence, even if robots are \NRigid, after finitely many turns at least one of them reaches the SEC, and therefore eventually they all reach the SEC. At this point, they correctly terminate.

The same circle formation phase can also be used in conjunction with the algorithms discussed in Section~\ref{s:asynch} for \ASynch robots. The only difference is that, instead of modifying the \ASynch algorithms like we did with the \SSynch ones, we simply add an extra state, called \emph{Done}, to synchronize robots and make them transition correctly from the vertex adjustments phase into the circle formation phase. That is, instead of terminating, a robot sets itself to \emph{Done}, and then waits until all other robots are set to \emph{Done}, as well. Only then does it proceed to executing the circle formation phase described above. Of course, before the circle formation phase starts, if a robot sees another robot set to \emph{Done}, it treats its like an \emph{External} robot. This works with both \Rigid and \NRigid\ \ASynch robots carrying 4-colored lights.

\subsection{Converging to a Point Without Colliding}\label{s:nearg}
A simple modification of Algorithm~\ref{alg1} solves the \Near problem, which requires all the robots to converge to a point without colliding: it is sufficient to remove lines~8, 9, and~23, that is, all the operations involving colors, and the termination condition. Indeed, if there is only one internal robot, either it will become external, or the other robots will converge to its location. On the other hand, if all robots become external, the convex hull will keep shrinking until its vertices converge to a point. This works for \Rigid\ \SSynch robots, even without the use of colored lights.

However, if the robots carry 2-colored lights, they can also terminate when they get close enough to one another. This is done by simply modifying the termination condition of line~9:
$$\mbox{\textbf{if} } \forall r,s\in\mathcal V,\, r.\mbox{\emph{light}}=\mbox{\emph{Vertex}}\mbox{ \textbf{and} }\norm{r.\mbox{\emph{position}}-s.\mbox{\emph{position}}}<\varepsilon \mbox{ \textbf{then} Terminate}$$
where $\varepsilon$ is any given positive constant.

\subsection{\NRigid\ \SSynch Robots with Knowledge of $\delta$}\label{s:delta}
Suppose that the robots are \NRigid\ \SSynch, and as such they can be stopped by the scheduler at each turn before they reach their destination point, but not before they have moved by at least $\delta$. Recall that in this case they have an algorithm for \Obst that uses 3-colored lights, described in Section~\ref{s:ssynch2}. However, if the robots know the exact value of $\delta$ and they can use it in their computations, they can solve \Obst even with 2-colored lights, by executing a slightly modified version of Algorithm~\ref{alg1}.

If all the robots are initially collinear, Algorithm~\ref{alg1} makes them reach a non-collinear configuration, even if they are \NRigid. Subsequently, the invariants discussed in Section~\ref{s:invariants} keep holding, and in particular the convex hull of the robots never grows, and vertex robots never become non-vertex robots. We have to show that a version of Lemma~\ref{l:main} can be obtained for this model, as well. Referring to Figure~\ref{fig1}, we can make the robot in $p$ move toward $(a+b)/2$ by a smaller amount, never passing internal robots, and never colliding with them, unless they are closer than $\delta$. If an internal robot $r$ is closer than $\delta$ and it stands between $p$ and $(a+b)/2$, the robot in $p$ moves close enough to $r$, on the line parallel to $ab$, and it sets its light to the correct value (note that it knows before moving whether it will end up being a vertex robot or not). This ``lateral move'' cannot be stopped by the scheduler, and it is guaranteed to create a new external robot, and eventually increase by one the number of vertex robots.

On the other hand, if only ``non-lateral moves'' are made, the analysis in Section~\ref{s:convergence} can be generalized, because Equation~\ref{e:default} takes the form
$$r_i(t+1)=\frac \mu 2\cdot r_{i-1}(t)+\mu\cdot r_i(t)+\frac \mu 2 \cdot r_{i+1}(t),$$
where $\mu \in [\mu_0,1/2]$, and $\mu_0$ is a constant. Indeed, if the convex hull of the robots never grows, and its initial diameter is $d$, then each moving robot is guaranteed to move by at least a fraction of $\mu_0=\delta/d$ of its computed movement vector. Therefore, all the lemmas in Section~\ref{s:convergence} can be reproved by merely adjusting some coefficients in the formulas.

It remains to prove that, if only one internal robot is left, it eventually reaches the boundary of the convex hull without colliding with other robots. But since $\delta$ is known, we can make this robot stay still until it either becomes external (due to other robots' movements), or the diameter of the convex hull becomes smaller than $\delta$. As soon as it is guaranteed to make a reliable move, it can reach the midpoint of an edge of the convex hull, and therefore become external.

When all robots are external, they eventually reach a strictly convex configuration and they correctly terminate, as detailed in the proof of Theorem~\ref{t:correctness}.

\subsection{Trading Lights with the Knowledge of $n$}\label{s:knowledge}
Suppose that the robots do not carry visible lights and have no internal memory, but they share the knowledge of the total number of robots in the swarm, $n$. If the robots are \Rigid\ \SSynch, it is possible to slightly modify Algorithm~\ref{alg1} to solve \Obst in this model, as well.

Note that the information given by other robots' visible lights is used only when a robot has to terminate (line~9), or when it is the only internal robot and it has to move to the perimeter of the convex hull (line~23). However, both these situations can be recognized locally by counting the vertices of the convex hull of the visible robots: if it has $n$ non-degenerate vertices, \Obst has been solved, and the executing robot can terminate. If the convex hull has $n-1$ vertices and the executing robot is internal, it moves to the boundary, as in line~23 of Algorithm~\ref{alg1}.

The same techniques can be used to modify the algorithm of Section~\ref{s:delta}, so that \NRigid\ \SSynch robots with knowledge of $\delta$ and knowledge of $n$ can solve \Obst without the use of colored lights.

We are also able to optimize Algorithm~\ref{alg2} for robots with knowledge of $n$: namely, we can achieve termination detection even if the robots do not use the \emph{Adjusting} color, as follows. When all robots are external and a vertex robot makes a default move, it does not change its color, but remains \emph{External}. Then, when a vertex robot sees $n$ robots, it terminates. Note that making a default move allows a robot to see all other robots at its next activation, and therefore each external robot makes at most one default move before terminating. Moreover, when all robots are collinear, we apply this simple protocol: if a robot is an endpoint of the convex hull, it moves orthogonally to it (without changing color); otherwise it stays still. This way, as soon as an endpoint is activated, the configuration becomes non-collinear. The only exception to this rule is the case $n=3$, in which the central robot has to move orthogonally to the segment, while the other two robots stay still.

This technique allows \NRigid\ \SSynch robots with knowledge of $n$ to solve \Obst with 2 colors as opposed to 3.


\subsection{Fault Tolerance}\label{s:fault}
Observe that Lemma~\ref{l:vessels} requires only $n-1$ valves to be opened infinitely often, as opposed to $n$. This implies that, if \Rigid\ \SSynch robots execute the modification of Algorithm~\ref{alg1} described in Section~\ref{s:nearg}, they converge to a point even if one robot is unable to move. Therefore, in the presence of one faulty robot, \Near is still solvable, even without the use of colored lights. (On the other hand, if two robots are faulty, \Near is clearly unsolvable, because the two faulty robots cannot approach each other.) Additionally, if $n$ is known, \Obst and \Conv are solved in the presence of a faulty robot, provided that it is located on the boundary of the convex hull.

\subsection{Sequential Scheduler}
Suppose that the scheduler is \emph{sequential}, i.e., it is \SSynch and it activates exactly one robot at each turn. In this very special model there is a simple algorithm to solve \Obst with no termination detection, even with no colored lights and no knowledge of $n$, and even if the robots are \NRigid and two of them are faulty. (If three robots are faulty, \Obst is clearly unsolvable.) When a robot is activated, it just moves by a small amount, without crossing or landing on any line that passes through two robots that it can currently see (including itself), as illustrated in Figure~\ref{fig8}. Clearly, when a robot moves as described, it becomes visible to all other robots. Hence, when all robots (except possibly two of them) have moved at least once, they can all see each other.

\begin{figure}[h]
\centering
\includegraphics[width=.8\linewidth]{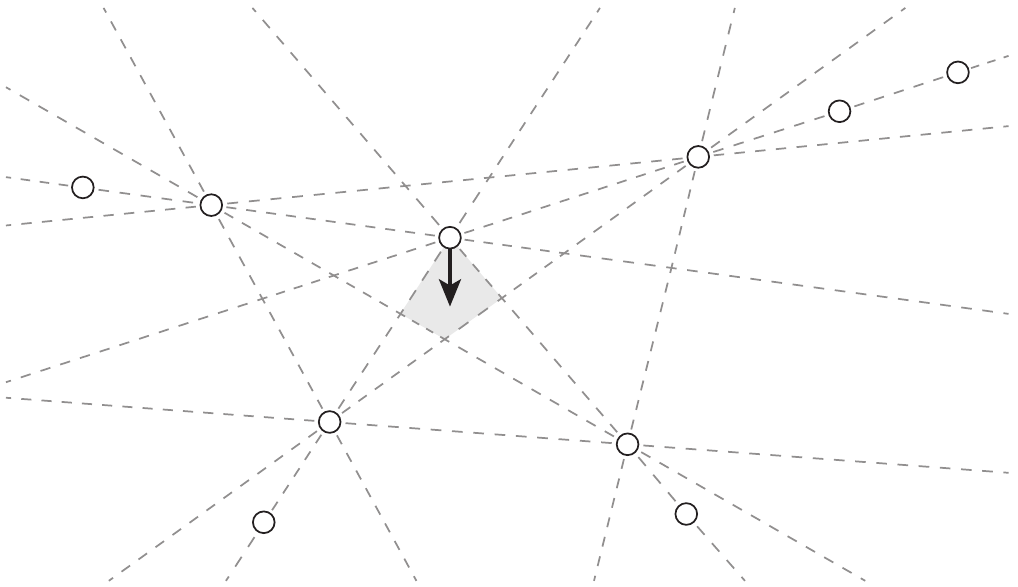}
\caption{Solving \Obst under a sequential scheduler}
\label{fig8}
\end{figure}

This protocol solves \Obst with no termination detection, in the sense that, after finitely many turns, the robots will keep seeing each other. However they will never stop moving because they will never know if their task is terminated or not. Indeed, termination detection is not achievable under this set of assumptions and, to be able to obtain it, some other assumptions are needed; for example,   2-colored lights  or the knowledge of $n$. With 2-colored lights, a robot can change its own color the first time it moves, and terminate at the next activation. With knowledge of $n$, a robot simply terminates when it sees $n$ robots.

\section{Concluding Remarks}\label{s:conclusions}
In this paper we initiated the investigation of the computational capabilities of a swarm of anonymous mobile robots with obstructed visibility: in this model, which has never been considered in the literature, two robots cannot see each other if a third robot lies between them. We focused on the basic problem of \Obst, in which the robots, starting from an arbitrary configuration, have to reach a configuration in which they all see each other, and then terminate the execution. This task is clearly impossible if the robots are completely oblivious, unable to communicate, and do not have any additional information. Indeed, in this scenario a robot can never distinguish between an initial configuration in which it cannot see some other robots, and a configuration in which all robots are visible and it is safe to terminate (recall that the termination operation cannot be undone).\footnote{It is worth noting that, if robots are only required to remain still forever after they have all become mutually visible (as opposed to terminating their execution), then this argument is no longer valid. With such a notion of \emph{weak termination}, there could exist an algorithm for \Obst that uses no colored lights and no extra information.} Therefore we employed the extended model of luminous robots, in which each robot is carrying a visible light that it can set to different colors. The goal is then to minimize the number of colors required by the robots to solve the \Obst problem under different settings and restrictions. Namely, we considered \SSynch and \ASynch robots, and \Rigid and \NRigid movements. We also discussed how to reduce the number of used colors if some information is given to the robots, such as the size of the swarm, $n$, or a minimum distance $\delta$ that a robot is guaranteed to cover in each movement. Our main results are summarized in Theorems~\ref{t:sum1} and~\ref{t:sum2}. We then touched on more complex problems, and proposed solutions that use our \Obst protocols as a preprocessing step. Notably, we gave the first algorithms for the \NearG problem (with fault tolerance) and the \Circ problem that work under the obstructed visibility model.

We proposed two main algorithms, and several modifications and adaptations to various models. Algorithm~\ref{alg1} (\emph{Shrink}) uses 2 colors and makes the convex hull of the robots shrink, ideally converging to a point. Algorithm~\ref{alg2} (\emph{Contain}) uses 3 colors, and keeps the initial convex hull basically unaltered. It is therefore suited for practical situations in which the robots have to surround a large-enough area, as well as solving \Obst. Also, both algorithms keep the robots within the initial convex hull (unless they are initially collinear), which is useful in practice, for instance in the presence of hazardous areas around the swarm.

Some interesting research problems remain unsolved. We would like to reduce the number of colors used by our algorithms in the various models, or to prove our algorithms optimal. Our solutions to \Obst in some models use only 2 colors (or no lights at all if $n$ is known), which is clearly optimal. For other models, such as \NRigid\ \SSynch and \Rigid\ \ASynch, we used 3 colors, and our question is whether this can be improved. We conjecture that Algorithm~\ref{alg1}, which uses only 2 colors and has been designed and proven correct for \Rigid\ \SSynch robots, can be applied with no changes also to \NRigid\ \SSynch robots (we could prove that 2 colors are sufficient in this model under the assumption that the robots know~$\delta$). In the \NRigid\ \ASynch setting we were only able to solve \Obst (with 3 colors) assuming that the robots agree on the direction of one coordinate axis. We ask if this assumption can be dropped, perhaps if more colors are used. Another question is whether \Obst can be solved deterministically without using colored lights or extra information, and without termination detection. That is, we allow the robots to move forever, but we require them to remain mutually visible from a certain time on. We proposed a simple solution that works under the sequential scheduler, and we ask if this can be generalized to \SSynch or even \ASynch schedulers.

We emphasize that obstructed visibility represents a broad line of research in the field of computation by mobile robots, and this paper explored just a few directions. Several classical problems are worth studying under this model, such as the general {\sf Pattern Formation} problem, {\sf Flocking}, {\sf Scattering}, with or without bounded visibility, etc.

\paragraph{\bf Acknowledgements} This work has been supported in part by the National Science and Engineering Research Council of Canada, under Discovery Grants, and by Professor Flocchini's University Research Chair.

\paragraph{\bf   References}
 \bibliographystyle{plain}

\end{document}